\documentclass[[letterpaper,12pt]{amsart}

\usepackage[hmargin=2.0cm,vmargin=2.0cm]{geometry}
\usepackage{tikz}
\usetikzlibrary{patterns}
\usetikzlibrary{decorations.pathmorphing}

\usepackage{amsmath,amsfonts,amssymb}

\usepackage{hyperref}
\usepackage{nameref,zref-xr}

\newtheorem{theorem}{Theorem}[section]
\newtheorem{proposition}[theorem]{Proposition}
\newtheorem{lemma}[theorem]{Lemma}
\newtheorem{corollary}[theorem]{Corollary}
\newtheorem{conjecture}[theorem]{Conjecture}

\newtheorem{rmk}[theorem]{Remark}
\newtheorem{observation}[theorem]{Observation}

\theoremstyle{definition}

\newtheorem{definition}[theorem]{Definition}

\newtheorem{remark}[theorem]{Remark}

\renewcommand{\d}{\partial}

\newcommand{\CM}{\mathcal {M}}
\newcommand{\CB}{\mathcal {B}}
\newcommand{\K}{{K}}
\newcommand{\oR}{\mathcal {R}}
\newcommand{\oSR}{\mathcal {SR}}
\newcommand{\mbQ}{\mathbb{Q}}
\newcommand{\oZ}{\overline{Z}}
\newcommand{\id}{\mathrm{id}}
\newcommand{\mcH}{\mathcal H}

\newcommand{\mbR}{\mathbb R}
\newcommand{\R}{{\mathbb R}}
\newcommand{\mbZ}{\mathbb Z}
\newcommand{\Z}{\mathbb Z}
\newcommand{\<}{\left <}
\renewcommand{\>}{\right >}
\newcommand{\mbC}{\mathbb C}
\newcommand{\mbL}{\mathbb{L}}
\newcommand{\CL}{\mathbb{L}}

\newcommand{\tF}{\widetilde F}

\newcommand{\ttau}{\widetilde{\tau}}

\newcommand{\oz}{\overline z}
\newcommand{\mcM}{\mathcal M}
\newcommand{\oM}{\overline{\mathcal M}}
\newcommand{\cM}{\mathcal M}
\newcommand{\oCM}{\overline{\mathcal M}}
\newcommand{\Mat}{\mathrm{Mat}}

\newcommand{\diag}{\mathop{\mathrm{diag}}\nolimits}
\newcommand{\tr}{\mathop{\mathrm{tr}}\nolimits}
\renewcommand{\deg}{\mathop{\mathrm{deg}}\nolimits}
\newcommand{\Faces}{\mathop{\mathrm{Faces}}\nolimits}
\newcommand{\Edges}{\mathop{\mathrm{Edges}}\nolimits}

\newcommand{\Aut}{\mathop{\mathrm{Aut}}\nolimits}
\renewcommand{\Re}{\mathop{\mathrm{Re}}\nolimits}
\renewcommand{\Im}{\mathop{\mathrm{Im}}\nolimits}
\newcommand{\tcR}{\widetilde{\mathcal{R}}}
\newcommand{\ok}{\overline{k}}

\newcommand{\exc}{\mathrm{exc}}
\newcommand{\cR}{\mathcal{R}}
\newcommand{\oU}{\overline{U}}
\newcommand{\toSR}{\widetilde{\oSR}}
\newcommand{\toR}{\widetilde{\oR}}
\newcommand{\tW}{\widetilde{W}}

\def\Tr{{\rm tr}\,}

\numberwithin{equation}{section}

\title{Refined open intersection numbers and the Kontsevich-Penner matrix model}

\author{Alexander Alexandrov}
\address{A.~Alexandrov: \newline Center for Geometry and Physics, Institute for Basic Science (IBS), Pohang 37673, Republic of Korea,
\newline CRM, Universit\'e de Montr\'eal, Montr´eal, Canada,
\newline Department of Mathematics and Statistics, Concordia University, Montreal, Canada, 
\newline and ITEP, Moscow, Russian Federation}
\email{alexandrovsash@gmail.com}

\author{Alexandr Buryak}
\address{A.~Buryak:\newline Department of Mathematics, ETH Zurich, Switzerland}
\email{buryaksh@gmail.com}

\author{Ran J. Tessler}
\address{R.~J.~Tessler:\newline Institute for Theoretical Studies, ETH Zurich, Switzerland}
\email{ran.tessler@eth-its.ethz.ch}

\begin{document}

\begin{abstract}
A study of the intersection theory on the moduli space of Riemann surfaces with boundary was recently initiated in a work of R. Pandharipande, J.~P. Solomon and the third author, where they introduced open intersection numbers in genus~$0$. Their construction was later generalized to all genera by J.~P.~Solomon and the third author. 
In this paper we consider a refinement of the open intersection numbers by distinguishing contributions from surfaces with different numbers of boundary components, and we calculate all these numbers. We then construct a matrix model for the generating series of the refined open intersection numbers and conjecture that it is equivalent to the Kontsevich-Penner matrix model. An evidence for the conjecture is presented. Another refinement of the open intersection numbers, which describes the distribution of the boundary marked points on the boundary components, is also discussed.
\end{abstract}

\maketitle

\section{Introduction}
A compact Riemann surface is a compact connected complex manifold of dimension~$1$. Denote by~$\mcM_{g,l}$ the moduli space of all compact Riemann surfaces of genus~$g$ with~$l$ marked points. P.~Deligne and D.~Mumford defined a natural compactification $\mcM_{g,l}\subset\oM_{g,l}$ via stable curves (with possible nodal singularities) in~\cite{DM69}. The moduli space $\oM_{g,l}$ is a non-singular complex orbifold of dimension~$3g-3+l$. It is defined to be empty unless the stability condition
\begin{gather}\label{closed stability}
2g-2+l>0
\end{gather}
is satisfied. We refer the reader to~\cite{DM69,HM98} for the basic theory.

In his seminal paper~\cite{Wit91}, E.~Witten initiated new directions in the study of $\oM_{g,l}$. For each marking index~$i$ consider the cotangent line bundle $\mbL_i \rightarrow \oM_{g,l}$, whose fiber over a point $[\Sigma,z_1,\ldots,z_l]\in \oM_{g,l}$ is the complex cotangent space $T_{z_i}^*\Sigma$ of $\Sigma$ at $z_i$. Let $\psi_i\in H^2(\oM_{g,l},\mathbb{Q})$ denote the first Chern class of $\mbL_i$, and write
\begin{gather}\label{eq:products}
\<\tau_{a_1} \tau_{a_2} \cdots \tau_{a_l}\>_g^c:=\int_{\oM_{g,l}} \psi_1^{a_1} \psi_2^{a_2} \cdots \psi_l^{a_l}.
\end{gather}
The integral on the right-hand side of~\eqref{eq:products} is well-defined, when the stability condition~\eqref{closed stability} is satisfied, all the $a_i$ are non-negative integers and the dimension constraint $3g-3+l=\sum a_i$ holds. In all other cases $\<\prod\tau_{a_i}\>_g^c$ is defined to be zero. The intersection products~\eqref{eq:products} are often called descendent integrals or \emph{intersection numbers}. Let $t_i$, $i\geq 0$, be formal variables and let
\begin{align*}
&F^c(t_0,t_2,\ldots):=\sum_{g\ge 0}F_g^c(t_0,t_1,\ldots),\quad\text{where}\\
&F^c_g(t_0,t_1,\ldots):=\sum_{l\ge 1}\sum_{a_1,\ldots,a_l\ge 0}\<\tau_{a_1}\tau_{a_2}\cdots\tau_{a_l}\>^c_g\frac{\prod t_{a_i}}{l!}.
\end{align*}
The generating series $F^c$ is called the \emph{closed free energy}. The exponent $\tau^c:=\exp(F^c)$ is called the \emph{closed partition function}. Witten's conjecture~(\cite{Wit91}), proved by M.~Kontsevich~(\cite{Kon92}), says that the closed partition function~$\tau^c$ becomes a tau-function of the KdV hierarchy after the change of variables~$t_n=(2n+1)!!T_{2n+1}$. Integrability immediately follows \cite{KMMMZ92} from Kontsevich's matrix integral representation
\begin{gather}\label{eq:Kontsmm}
\left.\tau^c\right|_{T_k=\frac{1}{k}\tr\Lambda^{-k}}=c_{\Lambda,M}\int_{\mcH_M}e^{\frac{1}{6}\tr H^3-\frac{1}{2}\tr H^2\Lambda} dH,
\end{gather}
where one integrates over the space of Hermitian $M\times M$ matrices, $\Lambda=\diag(\lambda_1,\ldots,\lambda_M)$ is a diagonal matrix with positive real entries and
$$
c_{\Lambda,M}:=(2\pi)^{-\frac{M^2}{2}}\prod_{i=1}^M\sqrt{\lambda_i}\prod_{1\le i<j\le M}(\lambda_i+\lambda_j).
$$

In~\cite{PST14} the authors started to develop a parallel theory for Riemann surfaces with boundary. A Riemann surface with boundary is a connected $1$~dimensional complex manifold with finite positive number of circular boundaries, each with a holomorphic collar structure. A compact Riemann surface is not viewed here as a Riemann surface with boundary. Given a Riemann surface with boundary $(X,\d X)$, we can canonically construct a double via Schwarz reflection through the boundary. The double $D(X,\d X)$ of $(X,\d X)$ is a compact Riemann surface. The doubled genus of $(X,\d X)$ is defined to be the usual genus of $D(X,\d X)$. On a Riemann surface with boundary~$(X,\d X)$, we consider two types of marked points. The markings of interior type are points of $X\backslash\d X$. The markings of boundary type are points of $\d X$. Let $\mcM_{g,k,l}^\R$ denote the moduli space of Riemann surfaces with boundary of doubled genus $g$ with $k$ distinct boundary markings and $l$ distinct interior markings. The moduli space $\mcM_{g,k,l}^\R$ is defined to be empty unless the stability
condition
$$
2g-2+k+2l>0
$$
is satisfied. The moduli space $\mcM_{g,k,l}^\R$ may have several connected components depending upon the topology of $(X,\d X)$ and the cyclic orderings of the boundary markings. Foundational issues concerning the construction of~$\mcM_{g,k,l}^\R$ are addressed in~\cite{Liu02}. The moduli space $\mcM_{g,k,l}^\R$ is a real orbifold of real dimension $3g-3+k+2l$, it is in general not compact and may be not orientable when $g>0.$

Since interior marked points have well-defined cotangent spaces, there is no difficulty in defining the cotangent line bundles $\mbL_i\to\mcM_{g,k,l}^\R$ for each interior marking, $i=1,\ldots,l$. Naively, one may want to consider a descendent theory via integration of products of the first Chern classes $\psi_i=c_1(\mbL_i)\in H^2(\oM_{g,k,l}^\R,\mbQ)$ over a compactification $\oM_{g,k,l}^\R$ of $\mcM_{g,k,l}^\R$. Namely,
\begin{gather}\label{eq:open products}
\<\tau_{a_1}\tau_{a_2}\cdots\tau_{a_l}\sigma^k\>^o_g:=2^{-\frac{g+k-1}{2}}\int_{\oM_{g,k,l}^\R}\psi_1^{a_1}\psi_2^{a_2}\cdots\psi_l^{a_l},
\end{gather}
when
$$
2\sum a_i=3g-3+k+2l,
$$
and in all other cases $\<\tau_{a_1}\tau_{a_2}\cdots\tau_{a_l}\sigma^k\>^o_g:=0$. Note that, in particular, $g+k$ must always be odd in order to get non-zero numbers. The new insertion $\sigma$ corresponds to the addition of a boundary marking. The coefficient in front of the integral on the right-hand side of~\eqref{eq:open products} appears to be useful for the description of the new intersection numbers, that are called the \emph{open intersection numbers}, in terms of integrable systems.

In genus $0$ the moduli $\oM_{0,k,l}:=\oM_{0,k,l}^\R$ is canonically oriented for $k$ odd, and one can calculate an integral of the form $\int_{\oM_{0,k,l}}\psi_1^{a_1}\psi_2^{a_2}\cdots\psi_l^{a_l},$ given \emph{boundary conditions} for the line bundles~$\mbL_i.$ More precisely, given nowhere vanishing boundary conditions $s\in C^\infty(E\to\partial\oCM_{0,k,l}),$ for $E = \bigoplus \mbL_i^{\oplus a_i},$ one may define the integral (\ref{eq:open products}) by
\begin{gather}\label{eq:open products_0}
\<\tau_{a_1}\tau_{a_2}\cdots\tau_{a_l}\sigma^k\>^o_0:=2^{-\frac{k-1}{2}}\int_{\oM_{0,k,l}}e(E,s),
\end{gather}
where $e(E,s)$ is the \emph{relative Euler class}. The result depends on the boundary conditions.

In \cite{PST14} a family of boundary conditions, called \emph{canonical boundary conditions} for each bundle~$\mbL_i$ is constructed. It is proven that for a generic choice of canonical boundary conditions, $s_{ij}\in C^\infty_m(\mbL_i\to\partial\oM_{0,k,l}),~i\in[l],j\in[a_i]$, the boundary conditions $s=\bigoplus s_{ij}$ is nowhere vanishing along $\partial\oM_{0,k,l},$ assuming $2\sum a_i=3g-3+k+2l$. Here we use the notation~$[l]$ for a set $\{1,2,\ldots,l\}$ and the subscript~$m$ indicates that \emph{multi-valued} section, rather than sections, are used. It is then shown that any two generic choices of canonical boundary conditions give rise to the same integral~\eqref{eq:open products_0}. In \cite{PST14} all open intersection numbers for doubled genus $0$ were calculated, and the authors proposed a conjectural description of the open intersection numbers in all genera. Let~$s$ be a formal variable. Define
\begin{align*}
&F^o(t_0,t_1,\ldots,s):=\sum_{g\ge 0}F^o_g(t_0,t_1,\ldots,s),\quad\text{where}\\
&F^o_g(t_0,t_1,\ldots,s):=\sum_{k,l\ge 0}\sum_{a_1,\ldots,a_l\ge 0}\<\tau_{a_1}\cdots\tau_{a_l}\sigma^k\>^o_g\frac{s^k\prod t_{a_i}}{k!l!}.
\end{align*}
The generating series $F^o$ is called the \emph{open free energy} and the exponent $\tau^o:=\exp(F^o+F^c)$ is called the \emph{open partition function}. The conjecture of R.~Pandharipande, J.~P.~Solomon and the third author~(\cite{PST14}) says that the generating series~$F^o$ satisfies a certain system of partial differential equations that is called in~\cite{PST14} the open KdV equations.

In higher genus the construction of open intersection numbers needs some refinement. Firstly, the moduli space $\cM_{g,k,l}$ is in general non-orientable for $g>0.$ In order to overcome this issue, J.~P.~Solomon and the third author define \emph{graded spin surfaces}, which are open surfaces with a spin structure and some extra structure. In \cite{STa} the moduli of graded spin surfaces $\oM_{g,k,l}$ is defined and is proved to be canonically oriented. When $g=0$ it coincides with $\oM_{0,k,l}^\R.$ Canonical boundary conditions are then constructed for the line bundles~$\mbL_i,$ and again it is proven that one can define
\begin{gather}\label{eq:open products_g}
\<\tau_{a_1}\tau_{a_2}\cdots\tau_{a_l}\sigma^k\>^o_g:=2^{-\frac{g+k-1}{2}}\int_{\oM_{g,k,l}}e(E,s),
\end{gather}
where $e(E,s)$ is the relative Euler with respect to the canonical boundary conditions. As in $g=0,$ generic choices of canonical boundary conditions give rise to the same integrals.
It should be stressed that, although \cite{STa} has not appeared yet, the moduli and boundary conditions mentioned above are fully described in Section~2 of~\cite{Tes15}.

A combinatorial formula for the open intersection numbers was found in~\cite{Tes15}. The conjecture of R.~Pandharipande, J.~P.~Solomon and the third author was proved in~\cite{BT15}. Properties of the open free energy $F^o$ were intensively studied in~\cite{Ale15a,Ale15b,Ale16,Bur15,Bur16,Saf16a}. In particular, in~\cite{Bur15,Bur16} the second author introduced a formal power series $F^{o,ext}(t_0,t_1,\ldots,s_0,s_1,\ldots)$, where $s_0=s$ and $s_1,s_2,\ldots$ are new formal variables. The function~$F^{o,ext}$ is an extension of the open free energy~$F^o$,
$$
\left.F^{o,ext}\right|_{s_{\ge 1}=0}=F^o,
$$
and, therefore, it was called the \emph{extended open free energy}. The exponent $\tau^{o,ext}:=\exp(F^{o,ext}+F^c)$ was called the \emph{extended open partition function}. In~\cite{Bur15,Bur16} the new variables $s_i$, $i\ge 1$, appeared naturally from the point of view of integrable systems. The second author suggested to consider them as descendents of the boundary marked points. A geometric construction of the descendent theory for the boundary marked points, a derivation of the combinatorial formula for it, and a geometric proof of the conjecture of~\cite{Bur15} regarding the extended theory, will appear in~\cite{STb},\cite{Tes}. 

In~\cite{Bur16} the second author found a simple relation of the extended open partition function~$\tau^{o,ext}$ to the wave function of the Kontsevich-Witten tau-function. In~\cite{Ale15b} the first author proved that both extended open partition function and closed partition function belong to the same family of tau-functions, described by the matrix integrals of Kontsevich type. Namely, the Kontsevich-Penner integral
\begin{gather}\label{eq:KPmatrixmod}
\left.\tau_N\right|_{T_k=\frac{1}{k}\tr\Lambda^{-k}}:=c_{\Lambda,M}\int_{\mcH_M}e^{\frac{1}{6}\tr H^3-\frac{1}{2}\tr H^2\Lambda}\frac{\det^N\Lambda}{\det^N(\Lambda-H)}dH
\end{gather}
for $N=0$ coincides with Kontsevich's integral (\ref{eq:Kontsmm}). In~\cite{Ale15b} it was shown that for $N=1$ it describes the extended open partition function. From this matrix integral representation it immediately follows that the extended open partition function is a tau-function of the KP hierarchy, moreover, it is related to the closed partition function $\tau^c$ by equations of the modified KP hierarchy \cite{KMMM93}.  A full set of the Virasoro and W-constrains for the tau-function, described by the Kontsevich-Penner matrix integral (\ref{eq:KPmatrixmod}), was derived in~\cite{Ale15b} for arbitrary $N$. Later these constraints were described by the first author \cite{Ale16} in terms of the so-called free bosonic fields.

\subsection{Refined, very refined and extended refined open intersection numbers}

As we already discussed above, the moduli space $\mcM_{g,k,l}$ may have several components depending upon the topology of Riemann surface with boundary. For $b\ge 1$, denote by $\mcM_{g,k,l,b}$ the submoduli of $\mcM_{g,k,l}$ that consists of isomorphism classes of surfaces with boundary with $b$ boundary components. So we have the decomposition
$$
\mcM_{g,k,l}=\bigsqcup_{\substack{1\le b\le g+1\\b+g=1(\mathrm{mod}\,\,2)}}\mcM_{g,k,l,b}.
$$
We can decompose further. Let $P(k,b)$ be the set of unordered $b$-tuples of non-negative integers $\ok=(k_1,\ldots,k_b)$, $k_i\ge 0$, such that $\sum k_i=k$. For $\ok=(k_1,\ldots, k_b)\in P(k,b)$ let $\mcM_{g,\bar{k},l}\subset\mcM_{g,k,l,b}$ be the submoduli of graded smooth Riemann surfaces with boundary of genus~$g$, with~$l$ internal marked points,~$b$ boundary components and~$k$ boundary marked points distributed on the boundary components according to the $b$-tuple~$\ok$. Clearly,
\[\mcM_{g,k,l,b}=\bigsqcup_{\bar{k}\in P(k,b)}\mcM_{g,\bar{k},l}.\]
It is also easy to see that if we define $\oM_{g,k,l,b}$ as the closure of $\mcM_{g,k,l,b}$ in $\oM_{g,k,l}$ and $\oM_{g,\ok,l}$ as the closure of $\mcM_{g,\ok,l}$ in $\oM_{g,k,l,b}$, then
\begin{align*}
\oM_{g,k,l}=&\bigsqcup_{\substack{1\le b\le g+1\\b+g=1(\mathrm{mod}\,\,2)}}\oM_{g,k,l,b},\\
\oM_{g,k,l,b}=&\bigsqcup_{\bar{k}\in P(k,b)}\oM_{g,\bar{k},l}.
\end{align*}
In \cite{STa} the authors defined open intersection numbers over any connected component of the moduli space~$\oM_{g,k,l}$. To be precise, they proved the following result.
\begin{theorem}\label{thm:int_defined}
Let $a_1,\ldots,a_l,k$ be non-negative integers satisfying $2\sum a_i=3g-3+k+2l,$ and let $E=\sum_{i=1}^k\mbL_i^{\oplus a_i}.$ Then for any connected component $C$ of $\oCM_{g,k,l}$ there exist nowhere vanishing canonical boundary conditions $s$ in the sense of \cite{PST14},\cite{STa}. Thus one may define the integral $\int_Ce(E,s)$. Moreover, any two nowhere vanishing choices of the canonical boundary conditions give rise to the same integral.
\end{theorem}
The theorem allows us to define \emph{refined open intersection numbers} as the integrals of monomials in psi-classes over the components $\oM_{g,k,l,b}$ of $\oM_{g,k,l}$ and \emph{very refined open intersection numbers} as the corresponding integrals over the components $\oM_{g,\bar{k},l}$:
\begin{align}
\<\tau_{a_1}\tau_{a_2}\cdots\tau_{a_l}\sigma^k\>^o_{g,b}:=&2^{-\frac{g+k-1}{2}}\int_{\oM_{g,k,l,b}}e(E,s),\label{eq:refined open products}\\
\<\tau_{a_1}\tau_{a_2}\cdots\tau_{a_l}\sigma^{\bar{k}}\>^o_{g}:=&2^{-\frac{g+k-1}{2}}\int_{\oM_{g,\bar{k},l}}e(E,s),\label{eq:very refined open products}
\end{align}
where $a_1,\ldots,a_l,k,E$ are as in Theorem \ref{thm:int_defined} and $s$ is a nowhere vanishing canonical multisection. These new intersection numbers are rational numbers. Let $N$ be a positive integer. Introduce the \emph{refined open free energy}~$F^{o,N}$ by
$$
F^{o,N}(t_0,t_1,\ldots,s):=\sum_{\substack{g,k,l\ge 0\\b\ge 1}}\sum_{a_1,\ldots,a_l\ge 0}\<\tau_{a_1}\cdots\tau_{a_l}\sigma^k\>^o_{g,b}\frac{N^b s^k\prod t_{a_i}}{k!l!}.
$$
Clearly, $F^{o,1}=F^o$. Let $q_0,q_1,\ldots$ be formal variables. Introduce the \emph{very refined open free energy}~$\tF^o$ by
$$
\tF^o(t_0,t_1,\ldots,q_0,q_1,\ldots):=\sum_{\substack{g,k,l\ge 0\\b\ge 1}}\sum_{\substack{\ok=(k_1,\ldots,k_b)\in P(k,b)\\a_1,\ldots,a_l\ge 0}}\<\tau_{a_1}\cdots\tau_{a_l}\sigma^{\ok}\>^o_g\frac{\prod t_{a_i}\prod q_{k_j}}{k!l!}.
$$
Of course, the function $F^{o,N}$ can be easily expressed in terms of the function $\tF^o$:
$$
F^{o,N}=\left.\tF^o\right|_{q_i=N s^i}.
$$
The reason, why we want to consider the refined open free energy $F^{o,N}$ separately, is that it admits a natural extension, while we do not know whether the very refined open free energy~$\tF^o$ can be extended. The exponents $\tau^o_N:=\exp(F^{o,N}+F^c)$ and~$\ttau^o:=\exp(\tF^o+F^c)$ will be called the \emph{refined open partition function} and the \emph{very refined open partition function} respectively.

In this paper we generalize the result of the third author from~\cite{Tes15} and find a combinatorial formula for the very refined open intersection numbers. We also derive matrix models for the refined and the very refined open partition functions. We then show that the form of our matrix model for the refined open partition function~$\tau^o_N$ suggests a natural way to add the variables $s_i$, $i\ge 1$, in it. We denote the extended function by~$\tau^{o,ext}_N$ and call it the \emph{extended refined open partition function}. This function satisfies the properties
$$
\left.\tau^{o,ext}_N\right|_{s_{\ge 1}=0}=\tau^o_N,\qquad \tau^{o,ext}_1=\tau^{o,ext}.
$$
Therefore, it is natural to view the variables $s_i$, $i\ge 1$, in the function $\tau^{o,ext}_N$ as descendents of the boundary marked points in the refined open intersection theory. We also prove that the extended refined open partition function~$\tau^{o,ext}_N$ is related to the very refined open partition function~$\ttau^o$ by a simple transformation. Moreover, we show that this transformation is invertible, so the collection of functions $\tau^{o,ext}_N$, $N\ge 1$, and the function $\ttau^o$ are in a certain sense equivalent. Finally, we conjecture that the function $\tau^{o,ext}_N$ coincides with the tau-function $\tau_N$ given by the Kontsevich-Penner matrix integral~\eqref{eq:KPmatrixmod} and present an evidence for the conjecture. In particular, we derive the string and the dilaton equations for the function $\tau^{o,ext}_N$ and also prove the conjecture in genus~$0$ and~$1$.

\begin{remark}
In~\cite{Saf16a} the author conjectured that there exists a refinement of the extended open partition function~$\tau^{o,ext}$ that distinguishes contributions from Riemann surfaces with different numbers of boundary components and that coincides with the Kontsevich-Penner tau-function~$\tau_N$. Since we construct this refinement, our conjecture can be considered as a stronger version of the conjecture of B.~Safnuk from~\cite{Saf16a}.
\end{remark}

\begin{remark}
Another approach to refined open intersection numbers was recently suggested by B.~Safnuk in~\cite{Saf16b}. His approach is quite different to ours, because, in particular, he does not consider boundary marked points and, moreover, he uses a different compactification of~$\mcM_{g,0,l}$. His intersection numbers are given as integrals of some specific volume forms. B.~Safnuk also has a combinatorial formula for his refined open intersection numbers and it directly gives the Kontsevich-Penner matrix model. It would be interesting to obtain a direct relation between the two approaches.
\end{remark}

\subsection{Organization of the paper}

In Section~\ref{section:refined numbers} we show that the construction of~\cite{STa} admits a refinement that allows to define the products~\eqref{eq:refined open products} and \eqref{eq:very refined open products}. We also prove combinatorial formulas for the refined  and the very refined open intersection numbers. In Section~\ref{section:matrix model} we construct a matrix model for the very refined open partition function~$\ttau^{o}_N$. We then show that the specialization of it, giving the refined open partition function, has a natural extension, where new variables can be interpreted as descendents of boundary marked points. We prove that the extended refined open partition function~$\tau^{o,ext}_N$ is related to the very refined open partition function by a simple transformation. We also prove the string and the dilaton equations for~$\tau^{o,ext}_N$. In Section~\ref{section:main conjecture} we formulate our conjecture about the relation between the function~$\tau^{o,ext}_N$ and the Kontsevich-Penner tau-function~$\tau_N$ and present an evidence for it.

\subsection{Acknowledgements}

We would like to thank Leonid Chekhov and Rahul Pandharipande for useful discussions. The work of A.A. was supported in part by IBS-R003-D1, by the Natural Sciences and Engineering Research Council of Canada (NSERC), by the Fonds de recherche du Qu\'ebec Nature et technologies (FRQNT) and by RFBR grants 15-01-04217 and 15-52-50041YaF.   A. B. was supported by Grant ERC-2012-AdG-320368-MCSK in the group of R.~Pandharipande at ETH Zurich and Grant RFFI-16-01-00409. R.T. is supported by Dr. Max R\"ossler, the Walter Haefner Foundation and the ETH Z\"urich
Foundation.


\section{Very refined open intersection numbers}\label{section:refined numbers}

\subsection{Reviewing the proof of the combinatorial formula of \cite{Tes15}}
In order to prove a combinatorial formula for the refined open intersection numbers, we first review the proof technique in the rather long paper \cite{Tes15}. Throughout this subsection we shall address to places in \cite{Tes15}.

\emph{Step 1}.
The starting point of \cite{Tes15} is the following well known fact. Let $M$ be an orbifold with boundary or even corners, of real dimension $2n.$ Suppose $E\to M$ is a vector bundle of real rank~$2n,$ and $s$ a nowhere vanishing (possibly multi-valued) section of $E\to\partial M.$ Let $\pi:S\to M$ be the sphere bundle associated to $E,~\Phi$ an angular form and $\Omega$ an Euler form on $M.$ In other words, $\Phi$ is a $2n-1$ form on the total space $S$ with
\begin{itemize}
\item
$\int_{\pi^{-1}(p)} \Phi = 1,~\forall p\in M$.
\item
$d\Phi = -\pi^*\Omega.$
\end{itemize}
Then we have
\begin{equation}\label{eq:Euler}\int_M e(E,s)=\int_M\Omega+\int_{\partial M}s^*\Phi.\end{equation}

\emph{Step 2}.
In \cite{Tes15}, Section $4,$ using the theory of Jenkins-Strebel differential \cite{Str84}, with the required modifications for graded surfaces with boundary, a combinatorial stratification of $\CM_{g,k,l}$ is constructed. The stratification, given a choice of positive perimeters $\mathbf{p}=\{p_1,\ldots, p_l\},$ consists of cells parameterized by metric graded ribbon graphs $(G,z).$ These are ribbon graphs with a (positive) metric on edges, $l+b$ holes, where the last $b$ holes, called \emph{boundaries} correspond to boundary components, the $i^{th}$ hole for $1\leq i\leq l$ is called a \emph{face} and is of perimeter $p_i,$ and there are $k$ boundary vertices which correspond to boundary marked points. $z$ is an index for the graded structure, whose description is not important at the moment. The topology of the cells is defined in the natural way using the metric. A cell $\CM_{(G',z')}$ is a face of a cell $\CM_{(G,z)}$ if $G'$ is obtained from $G$ by contracting some edges and $z'$ is the degenerated graded structure. The edge contraction operation allows a compactification of the combinatorial moduli, which is a quotient of $\oCM_{g,k,l},$ generically $1:1.$ Denote this compactification by $\oCM_{g,k,l}^{comb}(\mathbf{p}).$ Write also $\oCM_{g,k,l}^{comb}=\coprod_{p_1,\ldots, p_l>0}\oCM_{g,k,l}^{comb}(\mathbf{p}),$ and endow it with the natural topology and piecewise linear structure obtained by the graphs description. For later uses, write $\CM_{(G',z')}=\partial_e\CM_{(G,z)}$ if $(G',z')$ is the result of contracting the edge $e$ of $G.$

Not only the moduli, but also the $S^1$ bundles associated to the line bundles $\CL_i$ have a combinatorial counterpart, first obtained in \cite{Kon92}. Using these, in \cite{Tes15}, Subsection $4.3, $  a combinatorial~$S^{2n-1}$ bundle $S=S(E)$ is constructed for any vector bundle $E = \bigoplus \CL_i^{a_i},$ where $n=\sum a_i.$ It is then shown, in Proposition $4.39,$ that canonical multisections used to calculate the open intersection numbers can be taken to be pull backs of canonical multisections over $\oCM_{g,k,l}^{comb}.$ Call multisections of $S$ whose pull back is canonical \emph{combinatorial canonical}.
\cite{Tes15}, Lemma $4.42$ says
\begin{lemma}\label{lem:442}
For any $p_1,\ldots, p_l>0,$
\[\int_{\oCM_{g,k,l}}e(E,s) = \int_{\oCM_{g,k,l}^{comb}(\mathbf{p})}e(S,s'),\]
where $s$ is a canonical multisection which is a pull back of the combinatorial canonical multisection~$s'.$
\end{lemma}

\emph{Step 3}.
In \cite{Kon92} a combinatorial angular form $\alpha_i$ and a combinatorial curvature form $\omega_i$ were constructed, and using them a combinatorial formula for the closed numbers was obtained, by integration over highest dimensional cells, those parameterized by trivalent ribbon graph.
The main result of \cite{Tes15}, Section $3$ is an explicit formula for the angular form $\Phi$ of a bundle which is a direct sum of complex line bundles $L_i,$ in terms of their angular forms $\alpha_i$ and curvature forms~$\omega_i,$ such that $d\Phi$ is the pull back of $-\wedge \omega_i.$
Plugging this and (\ref{eq:Euler}) in Lemma \ref{lem:442} we get
\[2^{\frac{g+k-1}{2}}\langle\tau_{a_1}\cdots \tau_{a_l}\sigma^k\rangle_g^o = \int_{\oCM_{g,k,l}^{comb}}\bigwedge_{i=1}^l\omega_i^{a_i}+\int_{\partial\oCM_{g,k,l}^{comb}}(s')^*\Phi,\]
where $\Phi$ is the explicit angular form for $\bigoplus \CL_i^{\oplus a_i}.$

Finally, this equation can be simplified by noting that only highest dimensional cells of the combinatorial moduli and its boundary contribute to the integrals. The highest dimensional cells of $\oCM_{g,k,l}^{comb}$ are those parameterized by trivalent graded ribbon graphs. Denote their set by $\oSR_{g,k,l}^0.$ For any such graph, $(G,z)$ write $Br(G)$ for the set of \emph{bridges}, that is, edges which are either internal edges between two boundary vertices or boundary edges between boundary marked points. The highest dimensional cells in $\partial\oCM_{g,k,l}^{comb}$ are exactly those obtained from contracting a bridge in a cell of $\oSR_{g,k,l}^0.$ Putting all together we obtain (\cite{Tes15},Lemma 4.45)
\begin{equation}\label{eq:it0}
2^{\frac{g+k-1}{2}}\langle\tau_{a_1}\cdots \tau_{a_l}\sigma^k\rangle = \sum_{(G,z)\in\oSR_{g,k,l}^0}\int_{\CM_{(G,z)}(\mathbf{p})}\bigwedge_{i=1}^l\omega_i^{a_i}+
\sum_{\substack{(G,z)\in\oSR_{g,k,l}^0\\e\in Br(G)}}\int_{\CM_{\partial_e(G,z)}(\mathbf{p})}(s')^*\Phi,
\end{equation}
where $s'$ is combinatorial canonical.

\emph{Step 4}.
The expression (\ref{eq:it0}) has a complicated part, the integral of $(s')^*\Phi,$ since it involves the multisection $s'.$ However, it turns out that the properties of canonical sections allow computing the right-hand side of (\ref{eq:it0}) using iterative integrations by parts. The result is the integral version of the combinatorial formula. To this end, one must first have an explicit description of the contributing graded ribbon graphs.

\begin{definition}
Let $g,k,l$ be non-negative integers such that $2g-2+k+2l>0,~A$ be a finite set and $\alpha:[l]\to A$ a map. $\alpha,A$ will be implicit in the definition. A \emph{$(g,k,l)$-ribbon graph with boundary} is an
embedding $\iota:G\to\Sigma$ of a connected graph~$G$ into a $(g,k,l)$-surface with boundary~$\Sigma$ such that

\begin{itemize}

\item $\{x_i\}_{i\in [k]}\subseteq \iota(V(G))$, where $V(G)$ is the set of vertices of~$G$. We henceforth consider $\{x_i\}$ as vertices.

\item The degree of any vertex $v\in V(G)\setminus\{x_i\}$ is at least $3$.

\item $\partial\Sigma\subseteq \iota(G)$.

\item If $l \geq 1,$ then $$\Sigma\setminus\iota(G)=\coprod_{i\in [l]} D_i,$$ where each $D_i$ is a topological open disk, with $z_i\in D_i$. We call the disks $D_i$ faces.

\item If $l=0$, then $\iota(G)=\partial\Sigma$.

\end{itemize}
The genus $g(G)$ of the graph~$G$ is the genus of $\Sigma$. The number of the boundary components of~$G$ or~$\Sigma$ is denoted by $b(G)$ and $v_I(G)$ stands for the number of the internal vertices. We denote by~$\Faces(G)$ the set of faces of the graph~$G,$ and we consider $\alpha$ as a map $$\alpha\colon\Faces(G)\to A,$$
by defining for $f\in\Faces(G),~\alpha(f):=\alpha(i),$ where $z_i$ is the unique internal marked point in~$f.$ The map $\alpha$ is called the labeling of $G.$
Denote by~$V_{BM}(G)$ the set of boundary marked points~$\{x_i\}_{i\in [k]}.$

Two ribbon graphs with boundary $\iota\colon G\to\Sigma,~\iota'\colon G'\to\Sigma'$ are isomorphic, if there is an orientation preserving homeomorphism
$\Phi\colon(\Sigma,\{z_i\},\{x_i\})\to(\Sigma',\{z'_i\},\{x'_i\}),$ and an isomorphism of graphs $\phi\colon G\to G'$, such that
\begin{enumerate}
\item
$\iota'\circ\phi = \Phi\circ\iota.$
\item
$\phi(x_i)=x'_i,$ for all $i\in[k].$
\item
$\alpha'(\phi(f))=\alpha(f),$
where $\alpha,\alpha'$ are the labelings of $G,G'$ respectively and~$f\in\Faces(G)$ is any face of the graph~$G.$
\end{enumerate}
Note that in this definition we do not require the map~$\Phi$ to preserve the numbering of the internal marked points.

A ribbon graph is \emph{critical}, if
\begin{itemize}

\item Boundary marked points have degree $2$.

\item All other vertices have degree $3$.

\item If $l=0,$ then $g=0$ and $k=3.$

\end{itemize}
A $(0,3,0)-$ribbon graph with boundary is called a \emph{ghost}.

Consider maps $K$ from the set of directed edges of $G$ to $\Z_2$ which satisfy
\begin{itemize}
\item $K(e)+K(\bar{e})=1,$ where $\bar{e}$ is $e$ with opposite orientation.
\item For any face $f_i$ of the graph $G$ we have $\sum K(e)=1$, where the sum is taken over the directed edges of $f_i$, whose direction agree with the orientation of $f_i$.
\item Any directed edge of a boundary component has $K=0$.
\end{itemize}
A \emph{grading} of a critical ribbon graph is the equivalence class of such maps modulo the relations obtained by vertex flips. That is, $K,K'$ are identified if they differ by a sequence of moves which flip all the edge assignments for the edges which touch a vertex~$v$. Write $[K]$ for the equivalence class of $K.$ A graph together with a grading is called a \emph{graded graph}.

A \emph{metric graded graph} is a graded graph $(G,[K])$ together with a metric $\ell:\Edges(G)\to\R_+.$ Let $\CM_{(G,[K])}$ be the moduli of such metrics.
\end{definition}
From now on the explicit object $[K]$ will replace the abstract index $z$ used so far.

In Fig. \ref{fig:nonnodal}~two critical ribbon graphs are shown, the right one is a ghost.
We draw internal edges as thick (ribbon) lines, while boundary edges are usual lines.
Note that not all boundary vertices are boundary marked points.
We draw parallel lines inside the ghost, to emphasize that the face bounded by the boundary is a special face, without a marked point inside.
\begin{figure}[t]
\centering
\includegraphics[scale=.8]{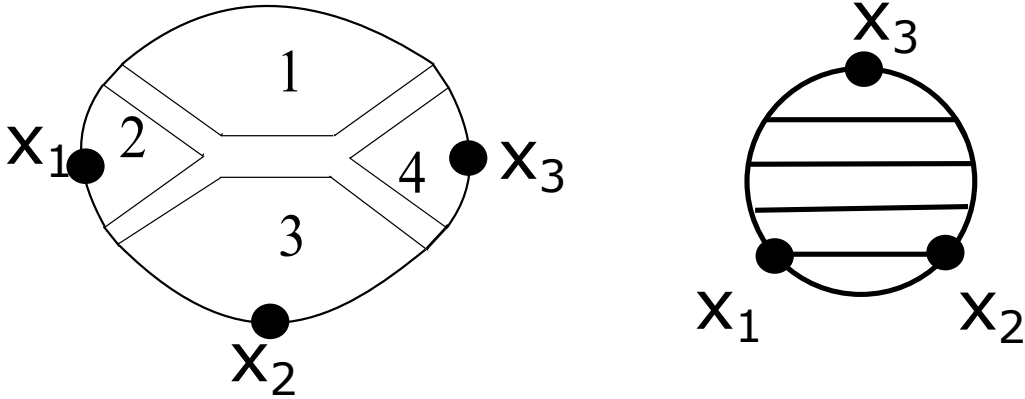}
\caption{Critical ribbon graphs.}
\label{fig:nonnodal}
\end{figure}

\begin{definition}
A \emph{nodal ribbon graph with boundary} is $G=\left(\coprod_i G_i\right)/N$, where
\begin{itemize}
\item $\iota_i\colon G_i\to\Sigma_i$ are ribbon graphs with boundary.
\item $N\subset (\cup_i V_{BM}(G_i))\times(\cup_i V_{BM}(G_i))$ is a set of \emph{ordered} pairs of boundary marked points $(v_1,v_2)$, $v_1\ne v_2$, of the $G_i$'s which we identify.
\end{itemize}
We require that
\begin{itemize}
\item $G$ is a connected graph,
\item Elements of $N$ are disjoint as sets (without ordering).
\end{itemize}

After the identification of the vertices~$v_1$ and~$v_2$ the corresponding point in the graph is called a node. The vertex~$v_1$ is called the legal side of the node and the vertex~$v_2$ is called the illegal side of the node.

The set of edges $\Edges(G)$ is composed of the internal edges of the $G_i$'s and of the boundary edges.
The boundary edges are the boundary segments between successive vertices which are not the illegal sides of nodes. For any boundary edge $e$ we denote by $m(e)$ the number of the
illegal sides of nodes lying on it. The boundary marked points of~$G$ are the boundary marked points of~$G_i$'s, which are not nodes. The set of boundary marked points of~$G$ will
be denoted by~$V_{BM}(G)$ also in the nodal case.

A nodal graph~$G=\left(\coprod_i G_i\right)/N$ is \emph{critical}, if
\begin{itemize}

\item All of its components~$G_i$ are critical.

\item Ghost components do not contain the illegal sides of nodes.

\end{itemize}
It is called \emph{odd critical} if it is critical and
any boundary component of $G_i$ has an odd number of points that are the boundary marked points or the legal sides of nodes.

A graded (odd) critical nodal graph $(G,[K])$ is a critical (odd) ribbon graph with gradings associated to each component $G_i.$

A nodal ribbon graph with boundary is naturally embedded into the nodal surface $\Sigma=\left(\coprod_i\Sigma_i\right)/N$. The (doubled) genus of $\Sigma$ is called the genus of the graph. The
notions of an isomorphism and metric are also as in the non-nodal case. Write $\CM_{(G,[K])}$ for the moduli of metrics on $(G,[K]).$
\end{definition}
\begin{rmk}
The genus of a closed, and in particular doubled, nodal surface $\Sigma$ is the genus of the smooth surface obtained by smoothing all nodes of $\Sigma.$
\end{rmk}

In Fig.~\ref{nodal}~there is a critical nodal graph of genus~$0$, with~$5$ boundary marked points, $6$ internal marked points, three components, one of them is a ghost, two nodes, where a
plus sign is drawn next to the legal side of a node and a minus sign is drawn next to the illegal side.

\begin{figure}[t]
\centering
\includegraphics[scale=.8]{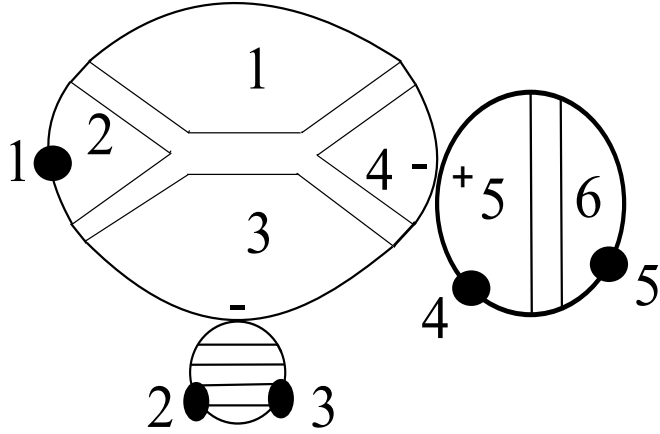}
\caption{A critical nodal ribbon graph.}
\label{nodal}
\end{figure}

In Fig.~\ref{noncriticalnodal}~a non-critical nodal graph is shown. Here there is some vertex of degree $4,$ the components do not satisfy the parity condition and the ghost component has an illegal node.
\begin{figure}[t]
\centering
\includegraphics[scale=.8]{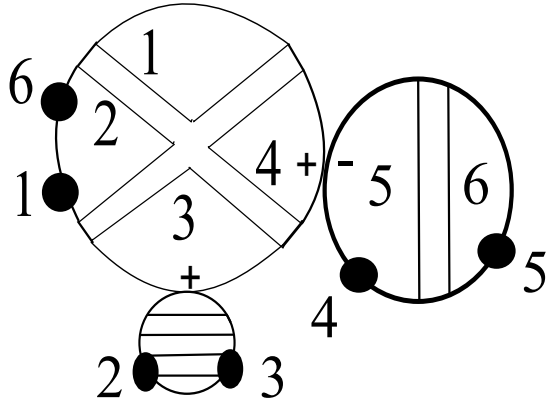}
\caption{A non-critical nodal ribbon graph.}
\label{noncriticalnodal}
\end{figure}

Let ${\oSR}^m_{g,k,l}(\toSR^m_{g,k,l})$ be the set of isomorphism classes of graded (odd) critical nodal ribbon graphs with boundary of genus~$g$, with~$k$ boundary marked points,~$l$ faces and together with a bijective labeling $\alpha:\Faces(G)\stackrel{\sim}{\to} [l]$, and $m$ nodes.

Denote by $\toR^m_{g,k,l}$ the set of isomorphism classes of odd critical nodal ribbon graphs with boundary of genus~$g$, with~$k$ boundary marked points,~$l$ faces and together with a bijective labeling $\alpha:\Faces(G)\stackrel{\sim}{\to} [l],$ and $m$ nodes. 

\begin{definition}
An \emph{effective bridge} in a graded critical graph $(G,[K])$ is a bridge $e$ with $m(e)=0.$ We denote their set by $Br^{eff}(G)$. The graph $\partial_e (G,[K])$, the result of contracting of the edge $e$ of $(G,[K]),$ which has one node $N$ more than $G$ has, can also be made critical nodal by declaring the side of $N$ which corresponds to $e$ to be legal, if $K(e)=0,$ and otherwise declare the other side of $N$ to be legal. Denote the resulting graph by $\CB\partial_e(G,[K]).$ The operation $\CB$ is called the \emph{base operation}.
\end{definition}

\begin{definition}
For a metric graded ribbon graph $G,$ define
\[W_G:=\prod_{e\in \Edges(G)}\frac{\ell_e^{2m(e)}}{(m(e)+1)!},\qquad\tW_G:=\prod_{e\in\Edges(G)}\frac{\ell_e^{2m(e)}}{m(e)!(m(e)+1)!}.\]
\end{definition}

\begin{definition}
An \emph{$l-$set} is a map $L:[n]\to [l].$ The \emph{size} of $L$ is $n.$
A \emph{subset} of an $l-$set is the restriction map $L:A\to[l],~A\subseteq[n].$ It can canonically identified with a map $L':[|A|]\to[l],$ hence can be thought as an $l-$set on its own right. We write $L'\subseteq L,$ and set $\binom{L}{m}$ for the set of all~$\binom{n}{m}~l-$subsets of $L$ of size $m.$
\end{definition}
\begin{definition}
Any $l-$set defines a vector bundle $E_L:=\bigoplus \CL_{L(i)},$ defined both on the moduli and on the combinatorial moduli. Let $S_L$ be the associated combinatorial sphere bundle. Let $\Phi_L$ be the associated explicit angular form, mentioned in Step $3$ above, and defined in~\cite[Section~3]{Tes15}. Its curvature form is $\omega_L = \bigwedge_{i\in[n]} \omega_{L(i)}.$
\end{definition}

\begin{lemma}\label{lem: for iterations}
Write $n = \frac{k+2l+3g-3}{2}$
Let $C\subseteq \oSR_{g,k,l}^m$ be a set of graphs and let $C'\subseteq \oSR_{g,k,l}^{m+1}$ be the set of graphs obtained by applying for any graph in $C$ and any effective bridge $e$ of it, first the edge contraction $\partial_e$ and then the base operation $\CB$. Suppose $C$ is closed in the following sense: for any graph $(G,[\K])$ in $\oSR_{g,k,l}^m\setminus C$ and any effective bridge $e$ of it we have $\CB(\partial_eG)\notin C'$. Then
\begin{align*}
\sum_{(G,[\K])\in C}\sum_{e\in Br^{eff}(G)}&\sum_{L'\in \binom{L}{n-m}}\int_{\CM_{\partial_e{(G,[\K])}}}W_G\Phi_{L'}=\\
&=\sum_{(G,[\K])\in C}\sum_{L'\in \binom{L}{n-m-1}}\left(\int_{\CM_{(G,[\K])}}W_G\omega_{L'} +\sum_{e\in Br^{eff}(G)}\int_{\CM_{\partial_e{(G,[\K])}}}W_G\Phi_{L'}\right).
\end{align*}
\end{lemma}
This lemma is the global version of the combination of Lemmas 6.7 and 6.8 of \cite{Tes15} (there a local version is given, in terms of a single graph, rather than a set $C,$ and in terms of a single $l-$subset of it, rather than summing over all subsets).

Applying Lemma \ref{lem: for iterations} iteratively to $C = \oSR_{g,k,l}^m,$ and using some parity observation (Proposition~6.13 in~\cite{Tes15}) give the integrated form of the combinatorial formula, \cite{Tes15}, Theorem~6.12.
\begin{theorem}\label{thm:int version}
For integers $a_1,\ldots, a_l\geq 0$ which sum to $n=\frac{k+2l+3g-3}{2},$ let $L$ be any $l-$set with $E_L=\bigoplus \CL_i^{\oplus a_i},$ then
\begin{gather*}
2^\frac{g+k-1}{2}\langle\tau_{a_1}\cdots\tau_{a_l}\sigma^k\rangle^o_{g}=\sum_{m\geq 0}\sum_{(G,[\K])\in\toSR_{g,k,l}^m}\sum_{L'\in\binom{L}{n-m}}\int_{\CM_{(G,[\K])}(\mathbf{p})}W_G\omega_{L'}.
\end{gather*}
\end{theorem}
A straightforward corollary is (equation (35) in \cite{Tes15})
\begin{corollary}\label{cor:int version}
\begin{gather*}
2^\frac{g+k-1}{2}\sum_{\sum_{i=1}^l a_i=n}\prod p_i^{2a_i}\langle\tau_{a_1}\cdots\tau_{a_l}\sigma^k\rangle^o_{g}=\sum_{m\geq 0}\sum_{(G,[\K])\in\toSR_{g,k,l}^m}\sum_{L'\in\binom{L}{n-m}}\int_{\CM_{(G,[\K])}(\mathbf{p})}\tW_G\frac{\bar\omega^{n-m}}{(n-m)!},
\end{gather*}
where $\bar{\omega} = \sum_i p_i^2\omega_i.$
\end{corollary}
Note that in the last theorem and corollary there is no more dependence on the choice of the multisection.

\emph{Step 5}.
The last step is to perform Laplace transform to the integrated formula described above. This is the content of \cite[Sections 6.2, 6.3]{Tes15}. The only difficulty in the calculation of the Laplace transform of
\[\int_{\CM_{(G,[\K])}(\mathbf{p})}\tW_G\frac{\bar\omega^{n-m}}{(n-m)!}\] for a given $(G,[K])\in\toSR^*_{g,k,l}$ is
to show
\[\bigwedge_{i=1}^ldp_i\wedge\frac{\bar\omega^{n-m}}{(n-m)!}:\bigwedge_{e\in\Edges(G)}d\ell_e=\pm \prod_i 2^{\frac{g(G_i)+b(G_i)-1}{2}+v_I(G_i)},\]
and to understand the signs. Here $G_i$ are the components of~$G$. This is the content of Section 6.2 in~\cite{Tes15}.

After understanding the sign and the ratio of forms, the Laplace transform calculations are straightforward and give
\begin{equation}\label{eq:local LP0}
\int_{p_i\in\mathbb{R}_+}\bigwedge dp_i\exp{\left(-\sum\lambda_i p_i\right)}\int_{\CM_{(G,[\K])}(\mathbf{p})}\tW_G\frac{\bar\omega^{n-m}}{(n-m)!}=
\pm\frac{\prod_i 2^{v_I(G_i)+\frac{g(G_i)+b(G_i)-1}{2}}}{|\Aut(G,[K])|}\prod_{e\in\Edges(G)}\lambda(e),
\end{equation}
where
\begin{gather}\label{eq:definition of lambda}
\lambda(e):=
\begin{cases}
\frac{1}{\lambda_i+\lambda_j},&\text{if $e$ is an internal edge between faces $i$ and $j$};\\
\frac{1}{(m+1)}\binom{2m}{m}\lambda_i^{-2m-1},&\text{if $e$ is a boundary edge of face $i$ and $m(e)=m$};\\
1,&\text{if $e$ is a boundary edge of a ghost}.
\end{cases}
\end{gather}
Summing over the different gradings $K,$ and using the results of Section 6.2 regarding the signs give
\begin{multline}\label{eq:local LP0}
\sum_{[K]\text{ is a grading for G}}\int_{p_i\in\mathbb{R}_+}\bigwedge dp_i\exp{\left(-\sum\lambda_i p_i\right)}\int_{\CM_{(G,[\K])}(\mathbf{p})}\tW_G\frac{\bar\omega^{n-m}}{(n-m)!}=\\
=\frac{\prod_i 2^{v_I(G_i)+\frac{g(G_i)+b(G_i)-1}{2}}}{|\Aut(G)|}\prod_{e\in\Edges(G)}\lambda(e).
\end{multline}

Summing over all graphs, the resulting combinatorial formula is
\begin{theorem}
Fix $g,k,l\ge 0$ such that $2g-2+k+2l>0$. Let $\lambda_1,\ldots,\lambda_l$ be formal variables. Then we have
\begin{multline}\label{eq:combinatorial formula}
2^{\frac{g+k-1}{2}}\sum_{a_1,\ldots,a_l\ge 0}\langle\tau_{a_1} \tau_{a_2} \cdots \tau_{a_l}\sigma^k\rangle_g^o\prod_{i=1}^l\frac{2^{a_i}(2a_i-1)!!}{\lambda_i^{2a_i+1}}=\\
=\sum_{G=\left(\coprod_i G_i\right)/N\in \widetilde{\oR}^*_{g,k,l}}\frac{\prod_i 2^{v_I(G_i)+g(G_i)+b(G_i)-1}}{|\Aut(G)|}\prod_{e\in\Edges(G)}\lambda(e).
\end{multline}
\end{theorem}

\subsection{A combinatorial formula for the refined and very refined numbers}

In order to write a combinatorial formula for the more refined numbers, first note
\begin{observation}\label{obs:nodal gives topology}
Let $(G',[K'])\in\oSR^m_{g,k,l}$ be an arbitrary graph, then there exists a graph $(G,[K])\in\oSR^0_{g,k,l},$ called the smoothing of $(G',[K'])$ and a sequence $(e_j)_{j=1}^m$ of bridges of $G$ such that
\[\CB\partial_{e_m}\cdots\CB\partial_{e_1}(G,[K])=(G',[K']).\]
Moreover, if $[\tilde{K}']$ is another graded structure on $G'$ then the smoothing of $(G',[\tilde{K}'])$ is some $(G,[\tilde{K}])$ with the same $G.$
Thus, the number of boundaries and partitions of boundary points of the smoothing of a graph $(G',[K'])$ is well-defined and independent of the graded structure.
\end{observation}
The proof is immediate, the operation $\CB$ remembers the cyclic order of the illegal nodes on each boundary edge, hence remembers the topology of the graph on which $\CB$ was applied. The edge contraction is easily inverted on the level of graphs, and the value of $K$ on the contracted bridge can be read from knowing which side of the node the $\CB$ operation declared to be illegal. The second part of the observation follows from the fact that the different gradings on $G'$ do not change the way we invert $\partial_e$.

Note that Steps 1--3 of the previous section work without change for the more refined numbers, giving us
\begin{equation}\label{eq:it0bdries}
2^{\frac{g+k-1}{2}}\langle\tau_{a_1}\cdots \tau_{a_l}\sigma^k\rangle^o_{g,b} = \sum_{(G,z)\in\oSR_{g,b,k,l}^0}\int_{\CM_{(G,z)}(\mathbf{p})}\bigwedge_{i=1}^l\omega_i^{a_i}+
\sum_{\substack{(G,z)\in\oSR_{g,b,k,l}^0\\e\in Br(G)}}\int_{\CM_{\partial_e(G,[K])}(\mathbf{p})}(s')^*\Phi,
\end{equation}
where $\oSR_{g,b,k,l}^m$ is the subset of $\oSR_{g,k,l}^m$ made of graphs whose smoothing has $b$ boundary components, and $s'$ is again combinatorial canonical. Define similarly $\toSR^m_{g,b,k,l},\oR^m_{g,b,k,l}$ and $\toR^m_{g,b,k,l}$. Define $\oSR_{g,\bar{k},l}^m,\toSR^m_{g,\ok,l},\oR^m_{g,\ok,l},\toR^m_{g,\ok,l}$, accordingly, for graphs which correspond to a partition $\bar{k}$ of boundary marked points. Then acting similarly for the very refined numbers yields
\begin{equation}\label{eq:it0bdries_very}
2^{\frac{g+k-1}{2}}\langle\tau_{a_1}\cdots \tau_{a_l}\sigma^{\bar{k}}\rangle^o_{g} = \sum_{(G,z)\in\oSR_{g,\bar{k},l}^0}\int_{\CM_{(G,z)}(\mathbf{p})}\bigwedge_{i=1}^l\omega_i^{a_i}+
\sum_{\substack{(G,z)\in\oSR_{g,\bar{k},l}^0\\e\in Br(G)}}\int_{\CM_{\partial_e(G,[K])}(\mathbf{p})}(s')^*\Phi,
\end{equation}
where $s'$ is again combinatorial canonical.

Step 4 requires some modification. Observation \ref{obs:nodal gives topology} allows us to apply Lemma \ref{lem: for iterations} to the sets~$C$ obtained by taking an arbitrary $(G,[K])\in\oSR^0_{g,k,l}$ and creating all elements of $\oSR^m_{g,k,l}$ obtained from it by contracting bridges and applying $\CB.$

Using Lemma \ref{lem: for iterations} iteratively now gives
\begin{theorem}\label{thm:refined_int version}
For integers $a_1,\ldots, a_l\geq 0$ which sum to $n=\frac{k+2l+3g-3}{2},$ let $L$ be any $l-$set with $E_L=\bigoplus \CL_i^{\oplus a_i},$ then
\begin{gather*}
2^\frac{g+k-1}{2}\langle\tau_{a_1}\cdots\tau_{a_l}\sigma^k\rangle^o_{g,b}=\sum_{m\geq 0}\sum_{(G,[\K])\in\toSR_{g,b,k,l}^m}\sum_{L'\in\binom{L}{n-m}}\int_{\CM_{(G,[\K])}(\mathbf{p})}W_G\omega_{L'},
\end{gather*}
and
\begin{gather*}
2^\frac{g+k-1}{2}\sum_{\sum_{i=1}^l a_i=n}\prod p_i^{2a_i}\langle\tau_{a_1}\cdots\tau_{a_l}\sigma^k\rangle^o_{g,b}=\sum_{m\geq 0}\sum_{(G,[\K])\in\toSR_{g,b,k,l}^m}\sum_{L'\in\binom{L}{n-m}}\int_{\CM_{(G,[\K])}(\mathbf{p})}\tW_G\frac{\bar\omega^{n-m}}{(n-m)!},
\end{gather*}
where $\bar{\omega} = \sum_i p_i^2\omega_i$.
\end{theorem}
Similarly, under the same assumptions,
\begin{theorem}\label{thm:very_refined_int version}
\begin{gather*}
2^\frac{g+k-1}{2}\langle\tau_{a_1}\cdots\tau_{a_l}\sigma^{\bar{k}}\rangle^o_{g}=\sum_{m\geq 0}\sum_{(G,[\K])\in\toSR_{g,\ok,l}^m}\sum_{L'\in\binom{L}{n-m}}\int_{\CM_{(G,[\K])}(\mathbf{p})}W_G\omega_{L'},
\end{gather*}
and
\begin{gather*}
2^\frac{g+k-1}{2}\sum_{\sum_{i=1}^l a_i=n}\prod p_i^{2a_i}\langle\tau_{a_1}\cdots\tau_{a_l}\sigma^{\ok}\rangle^o_{g}=\sum_{m\geq 0}\sum_{(G,[\K])\in\toSR_{g,\ok,l}^m}\sum_{L'\in\binom{L}{n-m}}\int_{\CM_{(G,[\K])}(\mathbf{p})}\tW_G\frac{\bar\omega^{n-m}}{(n-m)!}.
\end{gather*}
\end{theorem}

Step 5 follows without change, since the Laplace transform is performed cell-by-cell, and then summed over gradings, we see that for the refined numbers it holds that
\begin{theorem}\label{thm:comb refined}
Fix $g,k,l\ge 0$ such that $2g-2+k+2l>0$. Let $\lambda_1,\ldots,\lambda_l$ be formal variables. Then we have
\begin{multline}\label{eq:refined combinatorial formula}
2^{\frac{g+k-1}{2}}\sum_{a_1,\ldots,a_l\ge 0}\langle\tau_{a_1} \tau_{a_2} \cdots \tau_{a_l}\sigma^k\rangle_{g,b}^o\prod_{i=1}^l\frac{2^{a_i}(2a_i-1)!!}{\lambda_i^{2a_i+1}}=\\
=\sum_{G=\left(\coprod_i G_i\right)/N\in\widetilde{\oR}^*_{g,b,k,l}}\frac{\prod_i 2^{v_I(G_i)+g(G_i)+b(G_i)-1}}{|\Aut(G)|}\prod_{e\in\Edges(G)}\lambda(e).
\end{multline}
\begin{multline}\label{eq:very refined combinatorial formula}
2^{\frac{g+k-1}{2}}\sum_{a_1,\ldots,a_l\ge 0}\langle\tau_{a_1} \tau_{a_2} \cdots \tau_{a_l}\sigma^{\bar{k}}\rangle_{g}^o\prod_{i=1}^l\frac{2^{a_i}(2a_i-1)!!}{\lambda_i^{2a_i+1}}=\\
=\sum_{G=\left(\coprod_i G_i\right)/N\in\widetilde{\oR}^*_{g,\bar{k},l}}\frac{\prod_i 2^{v_I(G_i)+g(G_i)+b(G_i)-1}}{|\Aut(G)|}\prod_{e\in\Edges(G)}\lambda(e).
\end{multline}
\end{theorem}


\section{Matrix models}\label{section:matrix model}

In this section we present matrix models for the very refined and the extended refined open partition functions and study their properties. In Section~\ref{subsection:open partition function} we briefly recall the derivation of the matrix model for the open partition function~$\tau^o$. Then in Section~\ref{subsection:refined open partition function} we show how to modify it in order to control the distribution of boundary marked points on boundary components of a Riemann surface with boundary. As a result, we obtain a two-matrix model for the very refined open partition function~$\ttau^o$. In Section~\ref{subsection:extended refined open partition function} we give a construction of the extended refined open partition function~$\tau^{o,ext}_N$ and present simple transformations that relate it to the function~$\ttau^o$. In Section~\ref{subsection:Feynman for extended} we analyze the Feynman diagram expansion of the matrix integral for~$\tau^{o,ext}_N$ and then in Sections~\ref{subsection:string equation},~\ref{subsection:dilaton equation} derive the string and the dilaton equations for~$\tau^{o,ext}_N$.

It will be useful for the future to rewrite formula~\eqref{eq:very refined combinatorial formula} in the following way. For a graph $G=\left(\coprod_i G_i\right)/N\in\tcR^*_{g,\ok,l}$ introduce a combinatorial constant $c(G)$ by $c(G):=\prod_i c(G_i)$, where
\begin{gather}\label{eq:definition of constant}
c(G_i):=
\begin{cases}
\frac{1}{2},&\text{if $G_i$ is a ghost},\\
2^{e_I(G_i)-v_I(G_i)-v_{B3}(G_i)-v_{BM}(G_i)+b(G_i)},&\text{otherwise},
\end{cases}
\end{gather}
and $e_I(G_i)$ denotes the number of internal edges in $G_i$, $v_{B3}(G_i)$ is the number of boundary trivalent vertices and $v_{BM}(G_i)$ is the number of boundary marked points in $G_i$. Then for any $g,k,l\ge 0$, $b\ge 1$ and $\ok\in P(k,b)$ we have
\begin{gather}\label{eq:very refined combinatorial formula,2}
\sum_{a_1,\ldots,a_l\ge 0}\langle\tau_{a_1} \tau_{a_2}\cdots\tau_{a_l}\sigma^{\ok}\rangle_g^o\prod_{i=1}^l\frac{(2a_i-1)!!}{\lambda_i^{2a_i+1}}=\sum_{G=\left(\coprod_i G_i\right)/N\in\tcR^*_{g,\ok,l}}\frac{c(G)}{|\Aut(G)|}\prod_{e\in\Edges(G)}\lambda(e).
\end{gather}

\subsection{Open partition function}\label{subsection:open partition function}

Let $M\ge 1$. Consider positive real numbers $\lambda_1,\ldots,\lambda_M\in\mbR_{>0}$ and the diagonal matrix
$$
\Lambda:=\diag(\lambda_1,\ldots,\lambda_M).
$$
Let
$$
c_{\Lambda,M}:=(2\pi)^{-\frac{M^2}{2}}\prod_{i=1}^M\sqrt{\lambda_i}\prod_{1\le i<j\le M}(\lambda_i+\lambda_j).
$$
Denote by $\mcH_M$ the space of Hermitian $M\times M$ matrices. For a Hermitian matrix $H\in\mcH_M$ denote by $h_{i,j}$, $1\le i,j\le M$, its entries. Let
$$
t_i(\Lambda):=(2i-1)!!\tr\Lambda^{-2i-1},\quad i\ge 0.
$$
We consider the standard volume form
$$
dH:=\prod_{i=1}^M d h_{i,i} \prod_{1\leq i<j\leq M} d\left(\Re{h}_{i,j} \right)d \left(\Im{h}_{i,j}\right)
$$
on $\mcH_M$. In~\cite{BT15} the second and the third authors proved that
\begin{gather}\label{eq:matrix model for open}
\left.\tau^o\right|_{t_i=t_i(\Lambda)}=\left.e^{\frac{\d^2}{\d s\d s_-}}\left(e^{\frac{s^3}{6}}c_{\Lambda,M}\int_{\mcH_M}e^{\frac{1}{6}\tr H^3-\frac{1}{2}\tr H^2\Lambda}\det\frac{\Lambda+\sqrt{\Lambda^2-2s_-}-H+s}{\Lambda+\sqrt{\Lambda^2-2s_-}-H-s}dH\right)\right|_{s_-=0}.
\end{gather}
The integral in the brackets on the right-hand side of this expression can be understood in the sense of formal matrix integration. The form
$$
c_{\Lambda,M}e^{-\frac{1}{2}\tr H^2\Lambda}dH
$$
gives a Gaussian probability measure on~$\mcH_M$. Then we can expand the function
$$
e^{\frac{1}{6}\tr H^3}\det\frac{\Lambda+\sqrt{\Lambda^2-2s_-}-H+s}{\Lambda+\sqrt{\Lambda^2-2s_-}-H-s}
$$
in a series of the form
\begin{gather}\label{expression:series}
\sum_{a,b,m\ge 0}s^as_-^b P_{a,b,m},
\end{gather}
where $P_{a,b,m}$ is a polynomial of degree $m$ in expressions of the form $\tr(H\Lambda^{-d_1}H\Lambda^{-d_2}\cdots H\Lambda^{-d_r})$, $r\ge 1$. Here the degree is introduced by putting $\deg(\tr(H\Lambda^{-d_1}H\Lambda^{-d_2}\cdots H\Lambda^{-d_r})):=r+2\sum_{i=1}^r d_i$. Note that the integral
$$
c_{\Lambda,N}\int_{\mcH_N}P_{a,b,m}e^{-\frac{1}{2}\tr H^2\Lambda}dH
$$
is zero, if $m$ is odd, and is a rational function in $\lambda_1,\ldots,\lambda_N$ of degree $-\frac{m}{2}$, if $m$ is even. The integral on the right-hand side of~\eqref{eq:matrix model for open} is understood as the term-wise integral of~\eqref{expression:series} with respect to our Gaussian probability measure on~$\mcH_M$. We refer the reader to~\cite{BT15} for a more detailed discussion.

Let us briefly recall the derivation of formula~\eqref{eq:matrix model for open}. It is obtained from the combinatorial formula~\eqref{eq:combinatorial formula}, rewritten similarly to~\eqref{eq:very refined combinatorial formula,2}, using the standard matrix models technique. An odd critical nodal ribbon graph with boundary can be obtained from the disjoint union of critical non-nodal ribbon graphs with boundary by gluing boundary marked points. Since the sides of each node of the nodal graph are marked by plus or minus, we should assign pluses and minuses to the boundary marked points of the critical non-nodal ribbon graphs with boundary. A collar neighborhood of a boundary component of a critical non-nodal ribbon graph with boundary, that is not a ghost, is a circle with ribbon half-edges attached to it and also with boundary marked points (see Fig.~\ref{fig:boundary component}).
\begin{figure}[t]
\begin{tikzpicture}[scale=0.7]

\draw (0,0) circle (2.55);

\draw (2.53,0.3) -- (1,0.3);
\draw (2.53,-0.3) -- (1,-0.3);
\draw (-2.53,0.3) -- (-1,0.3);
\draw (-2.53,-0.3) -- (-1,-0.3);

\draw (0.3,2.53) -- (0.3,1);
\draw (-0.3,2.53) -- (-0.3,1);
\draw (0.3,-2.53) -- (0.3,-1);
\draw (-0.3,-2.53) -- (-0.3,-1);

\draw [thick,dotted] (0.3,1) -- (-0.3,1);
\draw [thick,dotted] (-1,0.3) -- (-1,-0.3);
\draw [thick,dotted] (0.3,-1) -- (-0.3,-1);
\draw [thick,dotted] (1,0.3) -- (1,-0.3);

\coordinate [label=center:\textbullet] (B) at (1.27,2.21);
\coordinate [label=-135:$+$] (B) at (1.27,2.21);
\coordinate [label=center:\textbullet] (B) at (1.80,-1.80);
\coordinate [label=135:$+$] (B) at (1.80,-1.80);
\coordinate [label=center:\textbullet] (B) at (-2.21,-1.27);
\coordinate [label=45:$-$] (B) at (-2.21,-1.27);

\draw [thick] (4,0) -- (5.5,0);
\draw [thick] (4,0.15) -- (4,-0.15);
\fill (5.25,0.15)--(5.5,0)--(5.25,-0.15)--cycle;

\coordinate [label=0:$\frac{s^2s_-}{2^6}\mathrm{tr}\left(H\Lambda^{-2}H\Lambda^{-2}H\Lambda^{-3}H\Lambda^{-1}\right)$] (B) at (6.5,0);

\end{tikzpicture}
\caption{Boundary piece}
\label{fig:boundary component}
\end{figure}
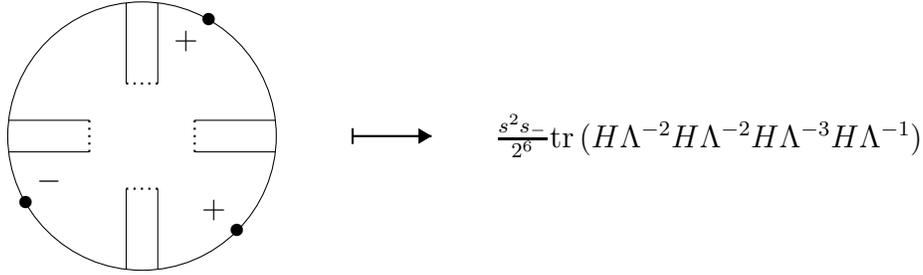
Such a circle with a configuration of ribbon half-edges and marked points will be called a boundary piece. We see that our odd critical nodal ribbon graph with boundary is obtained by
\begin{itemize}
\item gluing a set of trivalent stars (see Fig.~\ref{fig:internal and ghost})
\begin{figure}[t]
\begin{tikzpicture}

\begin{scope}[scale=0.43]

\draw (0.5,0) -- (0.5,-2.5);
\draw (-0.5,0) -- (-0.5,-2.5);
\draw (0.5,0) -- (2.67,1.25);
\draw (0,0.87) -- (2.17,2.12);
\draw (-0.5,0) -- (-2.67,1.25);
\draw (0,0.87) -- (-2.17,2.12);

\draw [thick,dotted] (0.5,-2.5) -- (-0.5,-2.5);
\draw [thick,dotted] (2.67,1.25) -- (2.17,2.12);
\draw [thick,dotted] (-2.67,1.25) -- (-2.17,2.12);

\end{scope}

\begin{scope}[shift={(-0.8,0)},scale=0.7]
\draw [thick] (4,0) -- (5.5,0);
\draw [thick] (4,0.15) -- (4,-0.15);
\fill (5.25,0.15)--(5.5,0)--(5.25,-0.15)--cycle;
\coordinate [label=0: $\frac{1}{6}\mathrm{tr}(H^3)$] (B) at (6.5,0);
\end{scope}

\begin{scope}[shift={(9,0)},scale=0.7]

\draw [pattern=horizontal lines, pattern color=black] (0,0) circle (1.5);
\coordinate [label=center:\textbullet] (B) at (0.87*1.5,-0.5*1.5);
\coordinate [label=center:\textbullet] (B) at (-0.87*1.5,-0.5*1.5);
\coordinate [label=center:\textbullet] (B) at (0,1.5);
\coordinate [label=center:$+$] (B) at (0,1.9);
\coordinate [label=center:$+$] (B) at (-0.87*1.5-0.3,-0.5*1.5-0.2);
\coordinate [label=center:$+$] (B) at (0.87*1.5+0.3,-0.5*1.5-0.2);

\draw [thick] (2.9,0) -- (4.4,0);
\draw [thick] (2.9,0.15) -- (2.9,-0.15);
\fill (4.15,0.15)--(4.4,0)--(4.15,-0.15)--cycle;
\coordinate [label=0: $\frac{s^3}{6}$] (B) at (5.5,0);

\end{scope}

\end{tikzpicture}
\caption{Trivalent star and a ghost}
\label{fig:internal and ghost}
\end{figure}
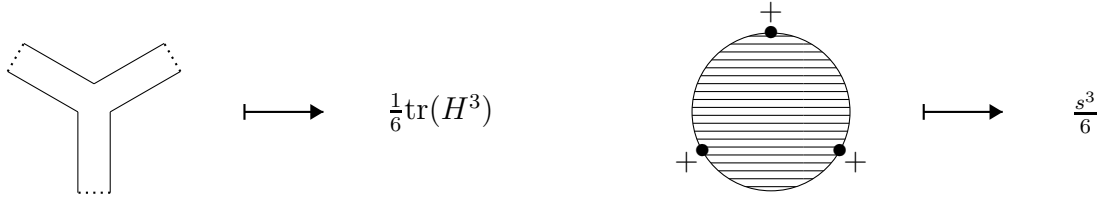
and boundary pieces,
\item taking the disjoint union with a number of ghost components (see Fig.~\ref{fig:internal and ghost}), and
\item gluing each boundary marked point coming with minus to a boundary marked point coming with plus.
\end{itemize}
Remember also that, according to the definition of an odd critical nodal ribbon graph with boundary, the number of boundary marked points coming plus on each boundary piece should be odd. We obtain that the trivalent stars give the contribution $e^{\frac{1}{6}\tr H^3}$ in the matrix model~\eqref{eq:matrix model for open}. Let
$$
G(\Lambda,s_-):=\sum_{m\ge 0}\frac{2^{-m}}{m+1}{2m\choose m}s_-^m\Lambda^{-2m-1}=\frac{2}{\Lambda+\sqrt{\Lambda^2-2s_-}}.
$$
Then boundary pieces give
\begin{gather*}
\exp\left(\tr\left[\sum_{k\ge 1}\frac{1}{k}\left(\frac{H+s}{2}G(\Lambda,s_-)\right)^k-\sum_{k\ge 1}\frac{1}{k}\left(\frac{H-s}{2}G(\Lambda,s_-)\right)^k\right]\right)=\det\frac{\Lambda+\sqrt{\Lambda^2-2s_-}-H+s}{\Lambda+\sqrt{\Lambda^2-2s_-}-H-s}.
\end{gather*}
The ghost components give the factor~$e^\frac{s^3}{6}$ in~\eqref{eq:matrix model for open}. The application of the operator $e^{\frac{\d^2}{\d s\d s_-}}$ and setting $s_-=0$ correspond to gluing each boundary marked point coming with minus to a boundary marked point coming with plus.

\subsection{Very refined open partition function}\label{subsection:refined open partition function}

Let us construct now a matrix model for the very refined open partition function $\ttau^o$. Let $N\ge 1$. In addition to the space~$\mcH_M$ of Hermitian matrices, we consider the space~$\Mat_{N,N}(\mbC)$ of complex $N\times N$ matrices. We consider it as a real vector space of dimension $2N^2$. For a matrix $Z\in\Mat_{N,N}(\mbC)$ denote by $z_{i,j}$, $1\le i,j\le N$, its entries. Define a volume form $dZ$ on $\Mat_{N,N}(\mbC)$ by
$$
dZ:=\prod_{1\le i,j\le N}d (\Re z_{i,j})d(\Im z_{i,j}).
$$
Consider the Gaussian probability measure on $\Mat_{N,N}(\mbC)$ given by the form
$$
\frac{1}{(2\pi)^{N^2}}e^{-\frac{1}{2}\tr Z\oZ^t}dZ.
$$
Let $\theta_{i,j}$, $1\le i,j\le N$, be complex variables and
\begin{align*}
\Theta:=&(\theta_{i,j})_{1\le i,j\le N}\in\Mat_{N,N}(\mbC),\\
q_m(\Theta):=&\tr\Theta^m,\quad m\ge 0.
\end{align*}
\begin{theorem}\label{theorem:very refined matrix model}
We have
\begin{align}
&\left.\ttau^o\right|_{\substack{t_i=t_i(\Lambda)\\q_i=q_i(\Theta)}}=\frac{c_{\Lambda,M}}{(2\pi)^{N^2}}\int_{\mcH_M\times\Mat_{N,N}(\mbC)}e^{-\frac{1}{2}\tr H^2\Lambda-\frac{1}{2}\tr Z\oZ^t}e^{\frac{1}{6}\tr H^3+\frac{1}{6}\tr Z^3+\frac{1}{2}\tr\oZ^t\Theta}\times\label{eq:very refined matrix model}\\
&\hspace{4cm}\times\det\frac{\Lambda\otimes\id_N+\sqrt{\Lambda^2\otimes\id_N-\id_M\otimes\oZ^t}-H\otimes\id_N+\id_M\otimes Z}{\Lambda\otimes\id_N+\sqrt{\Lambda^2\otimes\id_N-\id_M\otimes\oZ^t}-H\otimes\id_N-\id_M\otimes Z}dHdZ.\notag
\end{align}
\end{theorem}
\begin{proof}
We now use the combinatorial formula~\eqref{eq:very refined combinatorial formula,2}. As we explained in the previous section, an odd critical nodal ribbon graph with boundary is obtained by gluing trivalent stars (Fig.~\ref{fig:internal and ghost}) and boundary pieces (Fig.~\ref{fig:boundary component}), adding ghost components (Fig.~\ref{fig:internal and ghost}) and then gluing boundary points to create nodes. The problem now is to control the distribution of the boundary marked points on the boundary components in a smoothing the resulting nodal ribbon graph with boundary. Our idea is the following. Consider the nodal surface with boundary that is associated with our nodal ribbon graph with boundary. Consider a small neighborhood of a boundary node of this surface. At this node two small pieces of boundary components meet. Then, instead of gluing these two pieces at one point, we connect them by a small ribbon edge (see Fig.~\ref{fig:inserting external ribbon edges}).
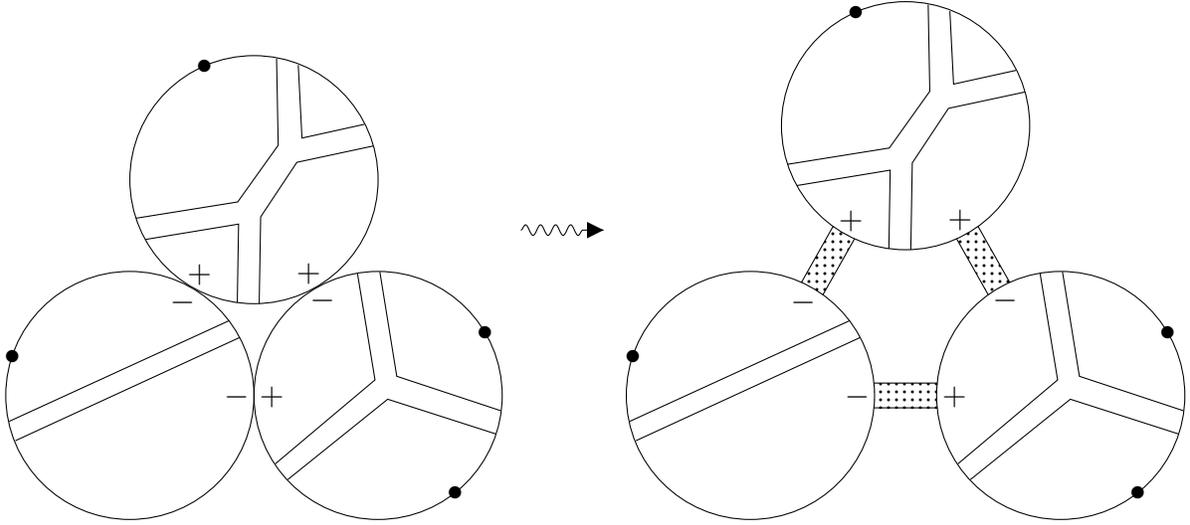
\begin{figure}[t]

\tikzset{snake it/.style={decoration={snake,
    amplitude = .7mm,
    segment length = 2mm,
    post length=2.2mm},decorate}}

\begin{tikzpicture}[scale=0.55]

\begin{scope}[shift={(9.45,4)}]
\path (0,0) edge[snake it] (2,0);
\fill (1.6,0.2)--(2,0)--(1.6,-0.2)--cycle;
\end{scope}

\draw (0,0) circle (3);
\draw (6,0) circle (3);
\draw (3,6*0.87) circle (3);

\draw (2.5965,2.2453) -- (2.6277,4.1449) -- (0.3771,3.7732);
\draw (0.149,4.2868) -- (2.6063,4.6806) -- (3.5853,6.0522) -- (3.5463,8.1464);
\draw (4.0704,7.9915) -- (4.1601,6.2261) -- (5.6662,6.5585);
\draw (3.1259,2.2219) -- (3.1619,4.3202) -- (4.0368,5.6396) -- (5.8819,6.0343);
\draw (3.515,-1.6547) -- (5.9255,0.376) -- (5.5019,2.9472);
\draw (6.0424,2.9911) -- (6.4515,0.4637) -- (8.9643,-0.3691);
\draw (3.8162,-2.0355) -- (6.2323,-0.0915) -- (8.8388,-0.9287);
\draw  (-2.9182,-0.6274) edge (2.3991,1.8093);
\draw  (-2.7674,-1.0946) edge (2.6589,1.4045);

\coordinate [label=center:\textbullet] (B) at (1.799,7.9549) {} {};
\coordinate [label=center:$+$] (B) at (1.6827,2.922) {};
\coordinate [label=center:$+$] (B) at (4.3186,2.9407) {} {};
\coordinate [label=center:$-$] (B) at (4.6582,2.3066) {} {} {};
\coordinate [label=center:$-$] (B) at (1.2759,2.2454) {} {} {} {};
\coordinate [label=center:$+$] (B) at (3.4291,-0.0454) {} {} {};
\coordinate [label=center:$-$] (B) at (2.5825,-0.0454) {} {} {} {};
\coordinate [label=center:\textbullet] (B) at (-2.8374,0.926) {} {} {};
\coordinate [label=center:\textbullet] (B) at (8.5799,1.505) {} {} {} {};
\coordinate [label=center:\textbullet] (B) at (7.8633,-2.3579) {} {} {} {} {};

\begin{scope}[shift={(15,0)}]

\draw (0,0) circle (3);
\draw  (-2.9182,-0.6274) edge (2.3991,1.8093);
\draw  (-2.7674,-1.0946) edge (2.6589,1.4045);

\begin{scope}[shift={(0.75,1.5*0.87)}]
\draw (3,6*0.87) circle (3);
\draw (2.5965,2.2453) -- (2.6277,4.1449) -- (0.3771,3.7732);
\draw (0.149,4.2868) -- (2.6063,4.6806) -- (3.5853,6.0522) -- (3.5463,8.1464);
\draw (4.0704,7.9915) -- (4.1601,6.2261) -- (5.6662,6.5585);
\draw (3.1259,2.2219) -- (3.1619,4.3202) -- (4.0368,5.6396) -- (5.8819,6.0343);
\coordinate [label=center:\textbullet] (B) at (1.799,7.9549) {} {};
\coordinate [label=center:$+$] (B) at (4.3186,2.9407) {} {};
\coordinate [label=center:$+$] (B) at (1.6827,2.922) {};
\end{scope}

\begin{scope}[shift={(1.5,0)}]
\draw (6,0) circle (3);
\draw (3.515,-1.6547) -- (5.9255,0.376) -- (5.5019,2.9472);
\draw (6.0424,2.9911) -- (6.4515,0.4637) -- (8.9643,-0.3691);
\draw (3.8162,-2.0355) -- (6.2323,-0.0915) -- (8.8388,-0.9287);
\coordinate [label=center:\textbullet] (B) at (8.5799,1.505) {} {} {} {};
\coordinate [label=center:\textbullet] (B) at (7.8633,-2.3579) {} {} {} {} {};
\coordinate [label=center:$-$] (B) at (4.6582,2.3066) {} {} {};
\coordinate [label=center:$+$] (B) at (3.4291,-0.0454) {} {} {};
\end{scope}

\coordinate [label=center:$-$] (B) at (1.2759,2.2454) {} {} {} {};
\coordinate [label=center:$-$] (B) at (2.5825,-0.0454) {} {} {} {};
\coordinate [label=center:\textbullet] (B) at (-2.8374,0.926) {} {} {};

\draw (1.23313,2.73485) --  (7.5/2-1.75188,7.5*0.87-2.43535);
\draw (1.75188,2.43535) -- (7.5/2-1.23313,7.5*0.87-2.73485);
\draw (2.98501,0.2995) -- (7.5-2.98501,0.2995);
\draw (7.5-2.98501,-0.2995) -- (2.98501,-0.2995);
\draw (7.5-1.23313,2.73485) --  (7.5-7.5/2+1.75188,7.5*0.87-2.43535);
\draw  (7.5-7.5/2+1.23313,7.5*0.87-2.73485)--(7.5-1.75188,2.43535);

\fill [pattern=dots, pattern color=black] (1.23313,2.73485) --  (7.5/2-1.75188,7.5*0.87-2.43535)-- (7.5/2-1.23313,7.5*0.87-2.73485)--(1.75188,2.43535)--cycle;
\fill [pattern=dots, pattern color=black] (2.98501,0.2995) -- (7.5-2.98501,0.2995) -- (7.5-2.98501,-0.2995) -- (2.98501,-0.2995) -- cycle;
\fill [pattern=dots, pattern color=black] (7.5-1.23313,2.73485) --  (7.5-7.5/2+1.75188,7.5*0.87-2.43535)-- (7.5-7.5/2+1.23313,7.5*0.87-2.73485)--(7.5-1.75188,2.43535)--cycle;

\end{scope}

\end{tikzpicture}
\caption{Inserting external ribbon edges}
\label{fig:inserting external ribbon edges}
\end{figure}
The new ribbon edge will be called an external ribbon edge. In Fig.~\ref{fig:inserting external ribbon edges} we fill the external ribbon edges by dots in order to distinguish them with the usual internal ribbon edges. Doing this procedure at each node, we obtain a non-nodal surface with boundary, which is a smoothing of the initial nodal surface. Note that each half of an external ribbon edge is marked by plus or minus.

Note that the resulting non-nodal surface can be glued from elementary pieces in the following way. Again we have trivalent stars~(Fig.~\ref{fig:internal and ghost}). Then we have boundary pieces similar to what we have in the previous section, but now we want to replace each boundary marked point by an external ribbon half-edge, marked by plus or minus (see Fig.~\ref{fig:boundary with external}).
\begin{figure}[t]
\begin{tikzpicture}[scale=0.7]

\draw (0,0) circle (2.55);

\draw (2.53,0.3) -- (1,0.3);
\draw (2.53,-0.3) -- (1,-0.3);
\draw (-2.53,0.3) -- (-1,0.3);
\draw (-2.53,-0.3) -- (-1,-0.3);

\draw (0.3,2.53) -- (0.3,1);
\draw (-0.3,2.53) -- (-0.3,1);
\draw (0.3,-2.53) -- (0.3,-1);
\draw (-0.3,-2.53) -- (-0.3,-1);

\draw [thick,dotted] (0.3,1) -- (-0.3,1);
\draw [thick,dotted] (-1,0.3) -- (-1,-0.3);
\draw [thick,dotted] (0.3,-1) -- (-0.3,-1);
\draw [thick,dotted] (1,0.3) -- (1,-0.3);

\coordinate [label=center:$+$] (B) at (1.1,1.9);
\draw (1.27-0.2,2.21+0.1) -- (2.02-0.2,3.51+0.1);
\draw (1.27+0.19,2.21-0.11) -- (2.02+0.19,3.51-0.11);
\fill [pattern=dots, pattern color=black] (1.27-0.2,2.21+0.1) -- (2.02-0.2,3.51+0.1) -- (2.02+0.19,3.51-0.11) --  (1.27+0.19,2.21-0.11) -- cycle;
\draw [thick,dotted] (2.02-0.2,3.51+0.1) -- (2.02+0.19,3.51-0.11);
\coordinate [label=center:$+$] (B) at (1.55,-1.6);
\draw (1.80+0.16,-1.80+0.16) -- (1.8+1.5/1.41+0.16,-1.8-1.5/1.41+0.16);
\draw (1.80-0.16,-1.80-0.16) -- (1.8+1.5/1.41-0.16,-1.8-1.5/1.41-0.16);
\fill [pattern=dots, pattern color=black] (1.80+0.16,-1.80+0.16) -- (1.8+1.5/1.41+0.16,-1.8-1.5/1.41+0.16) -- (1.8+1.5/1.41-0.16,-1.8-1.5/1.41-0.16) --  (1.80-0.16,-1.80-0.16) -- cycle;
\draw [thick,dotted] (1.8+1.5/1.41+0.16,-1.8-1.5/1.41+0.16) -- (1.8+1.5/1.41-0.16,-1.8-1.5/1.41-0.16);
\coordinate [label=center:$-$] (B) at (-1.9,-1.1);
\draw (-2.21-0.1,-1.27+0.2) -- (-3.51-0.1,-2.02+0.2);
\draw (-2.21+0.11,-1.27-0.19) -- (-3.51+0.11,-2.02-0.19);
\fill [pattern=dots, pattern color=black] (-2.21-0.1,-1.27+0.2) -- (-3.51-0.1,-2.02+0.2) --(-3.51+0.11,-2.02-0.19) --  (-2.21+0.11,-1.27-0.19) -- cycle;
\draw [thick,dotted] (-3.51-0.1,-2.02+0.2) -- (-3.51+0.11,-2.02-0.19);

\begin{scope}[shift={(-2,0)}]

\draw [thick] (6,0) -- (7.5,0);
\draw [thick] (6,0.15) -- (6,-0.15);
\fill (7.25,0.15)--(7.5,0)--(7.25,-0.15)--cycle;

\coordinate [label=0:$\frac{1}{2^7}\mathrm{tr}(H\Lambda^{-2}H\Lambda^{-2}H\Lambda^{-3}H\Lambda^{-1})\tr(Z^2\oZ^t)$] (B) at (8.5,0);

\end{scope}

\end{tikzpicture}
\caption{Boundary piece with external ribbon half-edges}
\label{fig:boundary with external}
\end{figure}
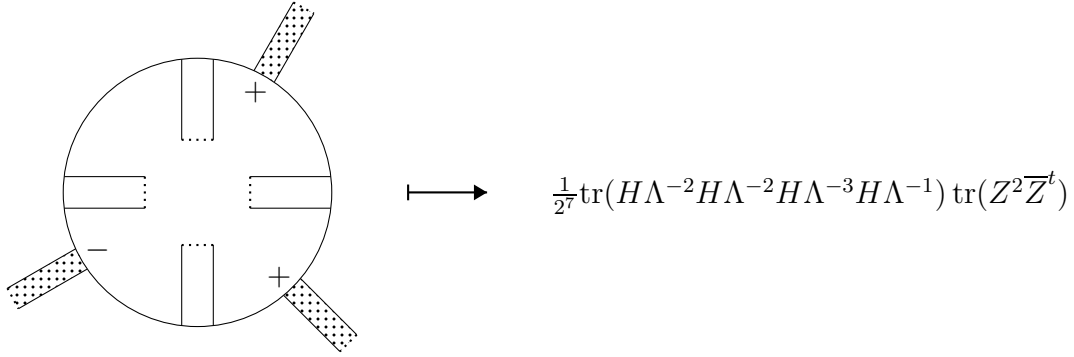
In the same way we replace the ghost component from Fig.~\ref{fig:internal and ghost} by the ghost component with external ribbon half-edges (see Fig.~\ref{fig:ghost with external}).
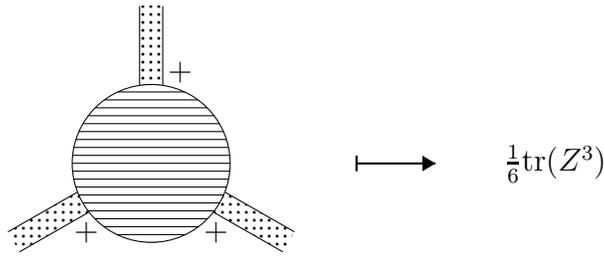
\begin{figure}[t]
\begin{tikzpicture}[scale=0.7]

\draw [pattern=horizontal lines, pattern color=black] (0,0) circle (1.5);
\coordinate [label=center:$+$] (B) at (0.55,1.73);
\coordinate [label=center:$+$] (B) at (-1.23,-1.3);
\coordinate [label=center:$+$] (B) at (1.23,-1.3);

\begin{scope}[shift={(1,0)}]
\draw [thick] (2.9,0) -- (4.4,0);
\draw [thick] (2.9,0.15) -- (2.9,-0.15);
\fill (4.15,0.15)--(4.4,0)--(4.15,-0.15)--cycle;
\coordinate [label=0: $\frac{1}{6}\mathrm{tr}(Z^3)$] (B) at (5.5,0);
\end{scope}

\draw (-0.22,1.49) -- (-0.22,2.99);
\draw (0.22,1.49) -- (0.22,2.99);
\fill [pattern=dots, pattern color=black] (-0.22,1.49) -- (-0.22,2.99) -- (0.22,2.99) --  (0.22,1.49) -- cycle;

\draw (1.31+0.1,-0.75+0.2) -- (2.62+0.1,-1.5+0.2);
\draw (1.31-0.12,-0.75-0.18) -- (2.62-0.12,-1.5-0.18);
\fill [pattern=dots, pattern color=black] (1.31+0.1,-0.75+0.2) -- (2.62+0.1,-1.5+0.2) -- (2.62-0.12,-1.5-0.18) -- (1.31-0.12,-0.75-0.18) -- cycle;

\draw (-1.31-0.1,-0.75+0.2) -- (-2.62-0.1,-1.5+0.2);
\draw (-1.31+0.12,-0.75-0.18) -- (-2.62+0.12,-1.5-0.18);
\fill [pattern=dots, pattern color=black] (-1.31-0.1,-0.75+0.2) -- (-2.62-0.1,-1.5+0.2) -- (-2.62+0.12,-1.5-0.18) -- (-1.31+0.12,-0.75-0.18) -- cycle;

\end{tikzpicture}
\caption{Ghost component with external ribbon half-edges}
\label{fig:ghost with external}
\end{figure}
In order to have marked points we have to introduce an external ribbon half-edge marked by minus (see Fig.~\ref{fig:external half-edge}).
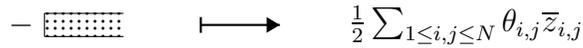
\begin{figure}[t]
\begin{tikzpicture}[scale=0.7]

\begin{scope}[shift={(0.07,0)}]
\draw [thick] (2.9,0) -- (4.4,0);
\draw [thick] (2.9,0.15) -- (2.9,-0.15);
\fill (4.15,0.15)--(4.4,0)--(4.15,-0.15)--cycle;
\coordinate [label=0: $\frac{1}{2}\sum_{1\le i,j\le N}\theta_{i,j}\oz_{i,j}$] (B) at (5.5,0);
\end{scope}

\draw (0,-0.22) -- (1.5,-0.22);
\draw (0,0.22) -- (1.5,0.22);
\draw (0,0.22) -- (0,-0.22);
\fill [pattern=dots, pattern color=black] (0,-0.22) -- (1.5,-0.22) -- (1.5,0.22) --  (0,0.22) -- cycle;

\coordinate [label=180: $-$] (B) at (0,0);

\end{tikzpicture}
\caption{External ribbon half-edge corresponding to a marked point}
\label{fig:external half-edge}
\end{figure}
Now, in order to obtain our non-nodal surface with external ribbon edges, we glue a set of elementary pieces of four types (Fig.~\ref{fig:internal and ghost},~\ref{fig:boundary with external},~\ref{fig:ghost with external},~\ref{fig:external half-edge}) according to the following rules:
\begin{itemize}
\item An internal ribbon half-edge should be glued to an internal ribbon half-edge.
\item An external ribbon half-edge with some sign should be glued to an external ribbon half-edge with an opposite sign.
\end{itemize}

For a polynomial $P(Z)\in\mbC[z_{ij},\oz_{kl}]$ let
$$
\<P(Z)\>:=\frac{1}{(2\pi)^{N^2}}\int P(Z)e^{-\frac{1}{2}\tr Z\oZ^t}dZ.
$$
Then we have
\begin{align}
&\<z_{i,j},z_{k,l}\>=\<\oz_{i,j},\oz_{k,l}\>=0,\label{eq:formula1 for complex}\\
&\<z_{i,j},\oz_{k,l}\>=2\delta_{i,k}\delta_{j,l}.\label{eq:formula2 for complex}
\end{align}
Formulas~\eqref{eq:formula1 for complex} and~\eqref{eq:formula2 for complex} show that our Gaussian probability measure on~$\Mat_{N,N}(\mbC)$ is the correct measure to control gluings of external ribbon half-edges with signs. To each elementary piece from Fig.~\ref{fig:internal and ghost},~\ref{fig:boundary with external},~\ref{fig:ghost with external},~\ref{fig:external half-edge} we assign a function on $\mcH_M\times\Mat_{N,N}(\mbC)$ in the way shown on these figures. Only the case of boundary pieces with external ribbon half-edges needs explanations. The function, corresponding to such a piece, is the product of a function on $\mcH_M$ and a function on $\Mat_{N,N}(\mbC)$. The function of $\mcH_{M}$ is obtained in the same way as in the previous section with the only difference that we forget about the variables~$s$ and~$s_-$. Concerning a function on~$\Mat_{N,N}(\mbC)$, we go around the boundary piece in the clockwise direction and look at the external ribbon half-edges that we meet. If an external ribbon half-edge is marked by plus then we assign to it the matrix~$Z$ and if it is marked by minus then we assign to it the matrix~$\oZ^t$. Then the function on $\Mat_{N,N}(\mbC)$ is the trace of the product of these matrices taken according to their order in the clockwise direction. So, the resulting function on $\mcH_M\times\Mat_{N,N}(\mbC)$ is the product of two traces. Note that the product of the traces of two matrices is the trace of their tensor product. We obtain that all boundary pieces with external ribbon half-edges give the following contribution to the matrix model for~$\ttau^o$:
\begin{gather}\label{expr:refined contribution from boundary}
\exp\left(\tr\left[\sum_{k\ge 1}\frac{1}{k}\left(\frac{H\otimes\id_N+\id_M\otimes Z}{2}G(\Lambda,\oZ^t)\right)^k-\sum_{k\ge 1}\frac{1}{k}\left(\frac{H\otimes\id_N-\id_M\otimes Z}{2}G(\Lambda,\oZ^t)\right)^k\right]\right),
\end{gather}
where $\id_M$ and $\id_N$ are the identity matrices in the spaces~$\mcH_M$ and~$\Mat_{N,N}(\mbC)$, respectively, and
$$
G(\Lambda,\oZ^t):=\sum_{m\ge 0}\frac{2^{-2m}}{m+1}{2m\choose m}\Lambda^{-2m-1}\otimes(\oZ^t)^m=\frac{2}{\Lambda\otimes\id_N+\sqrt{\Lambda^2\otimes\id_N-\id_M\otimes\oZ^t}}.
$$
We see that the expression~\eqref{expr:refined contribution from boundary} is equal to
$$
\det\frac{\Lambda\otimes\id_N+\sqrt{\Lambda^2\otimes\id_N-\id_M\otimes\oZ^t}-H\otimes\id_N+\id_M\otimes Z}{\Lambda\otimes\id_N+\sqrt{\Lambda^2\otimes\id_N-\id_M\otimes\oZ^t}-H\otimes\id_N-\id_M\otimes Z}.
$$
Finally, the trivalent stars give the contribution $e^{\frac{1}{6}\tr H^3}$ in the matrix model~\eqref{eq:very refined matrix model}, the ghost components with external ribbon half-edges give~$e^{\frac{1}{6}\tr Z^3}$ and the external ribbon half-edges corresponding to marked points give~$e^{\frac{1}{2}\tr\oZ^t\Theta}$. The theorem is proved.
\end{proof}

Using this theorem, we can obtain a matrix model for the refined open partition function~$\tau^o_N$ in the following way:
\begin{align}
&\left.\tau^o_N\right|_{t_i=t_i(\Lambda)}=\left.\left(\left.\ttau^o\right|_{\substack{t_i=t_i(\Lambda)\\q_i=q_i(\Theta)}}\right)\right|_{\Theta=s\,\id_N}=\label{eq:refined matrix model}\\
=&\frac{c_{\Lambda,M}}{(2\pi)^{N^2}}\int_{\mcH_M\times\Mat_{N,N}(\mbC)}e^{-\frac{1}{2}\tr H^2\Lambda-\frac{1}{2}\tr Z\oZ^t}e^{\frac{1}{6}\tr H^3+\frac{1}{6}\tr Z^3+\frac{s}{2}\tr\oZ^t}\times\notag\\
&\hspace{2cm}\times\det\frac{\Lambda\otimes\id_N+\sqrt{\Lambda^2\otimes\id_N-\id_M\otimes\oZ^t}-H\otimes\id_N+\id_M\otimes Z}{\Lambda\otimes\id_N+\sqrt{\Lambda^2\otimes\id_N-\id_M\otimes\oZ^t}-H\otimes\id_N-\id_M\otimes Z}dHdZ.\notag
\end{align}

\subsection{Extended refined open partition function}\label{subsection:extended refined open partition function}

The extended open partition function $\tau^{o,ext}\in\mbQ[[t_0,t_1,\ldots,s_0,s_1,\ldots]]$, introduced in~\cite{Bur15,Bur16}, is uniquely determined by the following equations:
\begin{align}
&\left.\tau^{o,ext}\right|_{s_{\ge 1}=0}=\tau^o,\label{eq:definition of extended,1}\\
&\frac{\d}{\d s_n}\tau^{o,ext}=\frac{1}{(n+1)!}\frac{\d^{n+1}}{\d s^{n+1}}\tau^{o,ext},\qquad n\ge 0.\label{eq:definition of extended,2}
\end{align}
Note that equation~(\ref{eq:refined matrix model}) gives a formula for $\tau^o$ that is slightly different to the initial formula~\eqref{eq:matrix model for open},
\begin{gather}
\label{eq:matrix model for open,2}
\left.\tau^{o}\right|_{t_i=t_i(\Lambda)}=\frac{c_{\Lambda,M}}{2\pi}\int_{\mcH_M\times\mbC}e^{-\frac{1}{2}\tr H^2\Lambda-\frac{1}{2}z\oz}e^{\frac{1}{6}\tr H^3+\frac{z^3}{6}+\frac{1}{2}s\oz}\det\frac{\Lambda+\sqrt{\Lambda^2-\oz}-H+z}{\Lambda+\sqrt{\Lambda^2-\oz}-H-z}dHd^2z,
\end{gather}
where $d^2z:=d(\Re z)d(\Im z)$. Formulas~\eqref{eq:definition of extended,1} and~\eqref{eq:definition of extended,2} imply that
\begin{align*}
&\left.\tau^{o,ext}\right|_{t_i=t_i(\Lambda)}=\\
=&\frac{c_{\Lambda,M}}{2\pi}\int_{\mcH_M\times\mbC}e^{-\frac{1}{2}\tr H^2\Lambda-\frac{z\oz}{2}}e^{\frac{1}{6}\tr H^3+\frac{z^3}{6}}\det\frac{\Lambda+\sqrt{\Lambda^2-\oz}-H+z}{\Lambda+\sqrt{\Lambda^2-\oz}-H-z}e^{\sum_{i\ge 0}\frac{2^{-i-1}}{(i+1)!}s_i\oz^{i+1}}dHd^2z.
\end{align*}
Let
$$
s_i(\Lambda):=2^i i!\tr\Lambda^{-2i-2},\quad i\ge 0.
$$
It is easy to see that
\begin{align*}
\left.e^{\sum_{i\ge 0}\frac{2^{-i-1}}{(i+1)!}s_i\oz^{i+1}}\right|_{s_i=s_i(\Lambda)}=&e^{\frac{1}{2}\sum_{i\ge 0}\frac{\oz^{i+1}}{(i+1)}\tr\Lambda^{-2i-2}}=e^{-\frac{1}{2}\tr\log(1-\oz\Lambda^{-2})}=\\
=&\det\frac{1}{\sqrt{1-\oz\Lambda^{-2}}}=\frac{\det\Lambda}{\det\sqrt{\Lambda^2-\oz}}.
\end{align*}
So, we get
\begin{align}\label{eq:extended matrix}
&\left.\tau^{o,ext}\right|_{\substack{t_i=t_i(\Lambda)\\s_i=s_i(\Lambda)}}=\\
=&\frac{c_{\Lambda,M}}{2\pi}\int_{\mcH_M\times\mbC}e^{-\frac{1}{2}\tr H^2\Lambda-\frac{z\oz}{2}}e^{\frac{1}{6}\tr H^3+\frac{z^3}{6}}\det\frac{\Lambda+\sqrt{\Lambda^2-\oz}-H+z}{\Lambda+\sqrt{\Lambda^2-\oz}-H-z}\frac{\det\Lambda}{\det\sqrt{\Lambda^2-\oz}}dHd^2z.\notag
\end{align}
This formula together with equation~\eqref{eq:refined matrix model} motivates us to introduce a formal power series $\tau^{o,ext}_N\in\mbC[[t_0,t_1,\ldots,s_0,s_1,\ldots]]$ by
\begin{align}\label{eq:full refined matrix model}
&\left.\tau^{o,ext}_N\right|_{\substack{t_i=t_i(\Lambda)\\s_i=s_i(\Lambda)}}=\frac{c_{\Lambda,M}}{(2\pi)^{N^2}}\int_{\mcH_M\times\Mat_{N,N}(\mbC)}e^{-\frac{1}{2}\tr H^2\Lambda-\frac{1}{2}\tr Z\oZ^t}e^{\frac{1}{6}\tr H^3+\frac{1}{6}\tr Z^3}\times\\
&\times\det\frac{\Lambda\otimes\id_N+\sqrt{\Lambda^2\otimes\id_N-\id_M\otimes\oZ^t}-H\otimes\id_N+\id_M\otimes Z}{\Lambda\otimes\id_N+\sqrt{\Lambda^2\otimes\id_N-\id_M\otimes\oZ^t}-H\otimes\id_N-\id_M\otimes Z}\frac{\det\Lambda^N dHdZ}{\det\sqrt{\Lambda^2\otimes\id_N-\id_M\otimes\oZ^t}}.\notag
\end{align}
The uniqueness of a power series with this property is obvious. However, the existence of such a series is not trivial. In order to prove it we will define a formal power series $\tau^{o,ext}_N$ using the function $\ttau^o$ and then prove that it satisfies equation~\eqref{eq:full refined matrix model}.

For a given $N\geq 1$ let us define a formal power series $\tau^{o,ext}_N\in\mbQ[[t_0,t_1,\ldots,s_0,s_1,\ldots]]$ by
\begin{gather}\label{eq:transformation}
\tau^{o,ext}_N(t_0,t_1,\ldots,s_0,s_1,\ldots):=\frac{1}{(2\pi)^{N^2}}\int_{\Mat_{N,N}(\mbC)}\left.\ttau^o\right|_{q_i=q_i(Z)}e^{\sum_{i\ge 0}\frac{2^{-i-1}}{(i+1)!}s_i\tr(\oZ^t)^{i+1}}e^{-\frac{1}{2}\tr Z\oZ^t}dZ.
\end{gather}
\begin{lemma}
The function $\tau^{o,ext}_N$ satisfies equation~\eqref{eq:full refined matrix model}.
\end{lemma}
\begin{proof}
Note that
\begin{gather}\label{eq:note}
\frac{\det\Lambda^N}{\det\sqrt{\Lambda^2\otimes\id_N-\id_M\otimes\oZ^t}}=\frac{1}{\det\sqrt{\id_M\otimes\id_N-\Lambda^{-2}\otimes\oZ^t}}=e^{\sum_{i\ge 0}\frac{2^{-i-1}}{(i+1)!}s_i(\Lambda)\tr(\oZ^t)^{i+1}}.
\end{gather}
Then the lemma follows from Theorem~\ref{theorem:very refined matrix model} and the elementary formula:
$$
\frac{1}{(2\pi)^{N^2}}\int_{\Mat_{N,N}(\mbC)}Q(\Theta)e^{\frac{1}{2}\tr\oZ^t\Theta}e^{-\frac{1}{2}\tr\Theta\overline{\Theta}^t}d\Theta=Q(Z),
$$
where $Q(\Theta)\in\mbC[\overline{\theta}_{i,j}]$ is an arbitrary polynomial.
\end{proof}

For a finite value of $N$ the transform, defined by the right hand side of (\ref{eq:transformation}) is not invertible. However, if we know $\tau^{o,ext}_N$ for all $N\geq 1$, we can find $\ttau^o$. Let us consider the space ${\mathcal U}_N$ of unitary $N\times N$ matrices. Then we introduce the volume form on ${\mathcal U}_N$, which is proportional to the Haar measure and normalized by
$$
\int_{{\mathcal U}_N} dU=1.
$$
Let $p_1,p_2,\ldots$ and $p_1',p_2',\ldots$ be formal variables.
\begin{lemma}
If
$$
f_N(p'_1,p'_2,\dots):=\frac{1}{(2\pi)^{N^2}}\int_{\Mat_{N,N}(\mbC)}\left.g \right|_{p_i=\frac{1}{i}\tr Z^i}e^{\sum_{i\ge 1}2^{-i} p'_i\tr(\oZ^t)^i}e^{-\frac{1}{2}\tr Z\oZ^t}dZ
$$
for some $g\in\mathbb{C}[[p_1,p_2,\dots]]$, then
$$
\left.g\right|_{p_i=\frac{1}{i}\tr A^i}=\int_{{\mathcal U}_N}  \left.f_N\right|_{p'_i=\frac{1}{i}\tr U^i}e^{\tr {\overline U}^tA}dU,\quad A\in\Mat_{N,N}(\mbC).
$$
\end{lemma}
\begin{proof}
The Schur functions $s_\lambda(p_1,p_2,\dots)$, labeled by partitions $\lambda=\{\lambda_1\geq\lambda_2\geq\lambda_3\geq\dots\}$, constitute a basis in the space of formal series in the variables $p_1,p_2,\dots$. Recall that they can be defined by
$$
s_\lambda:=\det(h_{\lambda_i-i+j})_{1\le i,j\le l(\lambda)},
$$
where the polynomials $h_k(p_1,p_2,\ldots)$, $k\in\mbZ$, are defined by
$$
\sum_{i\ge 0}h_i z^i=e^{\sum_{i\ge 1}p_i z^i}
$$
for $k\ge 0$ and by $h_k:=0$ for $k<0$. Thus, it is enough to prove the lemma for $g=s_\lambda$, where $l(\lambda)\le N$. (If $l(\lambda)> N$ then both $\left.g \right|_{p_i=\frac{1}{i}\tr Z^i}$ and $f_N$ are equal to zero.) For any matrix $A\in\Mat_{N,N}(\mbC)$ let
$$
p_i(A):=\frac{1}{i}\tr A^i,\quad i\ge 1.
$$
For any partition $\mu$ and matrices $A,B\in\Mat_{N,N}(\mbC)$ we have the following formula~\cite[eq.~(39)]{Ale11},
$$
\frac{1}{(2\pi)^{N^2}}\int_{\Mat_{N,N}(\mbC)}s_\lambda(p_*(ZA))s_\mu\left(p_*\left(\frac{1}{2}\oZ^tB\right)\right)e^{-\frac{1}{2}\tr Z\oZ^t}dZ=\frac{s_\lambda(p_*(AB))}{s_\lambda(1,0,0,\ldots)}\delta_{\lambda,\mu}.
$$
Then for any unitary matrix $U$ we can compute
\begin{align*}
&\frac{1}{(2\pi)^{N^2}}\int_{\Mat_{N,N}(\mbC)}s_\lambda(p_*(Z))e^{-\frac{1}{2}\tr Z\oZ^t+\sum_{i\ge 1} \frac{2^{-i}}{i}\tr U^i\tr(\oZ^t)^i}dZ=\\
=&\frac{1}{(2\pi)^{N^2}}\int_{\Mat_{N,N}(\mbC)}s_\lambda(p_*(Z))\sum_\mu s_\mu(p_*(U))s_\mu\left(p_*\left(1/2\oZ^t\right)\right)e^{-\frac{1}{2}\tr Z\oZ^t}dZ=\\
=&\frac{s_\lambda(p_*(\id_N))}{s_\lambda(1,0,0,\ldots)}s_\lambda(p_*(U))\stackrel{\text{\cite[Sec. 1.1]{Ale11}}}{=}C_N(\lambda)s_\lambda(p_*(U)),
\end{align*}
where
$$
C_N(\lambda)=
\prod_{i=1}^N\frac{(\lambda_i+N-i)!}{(N-i)!}.
$$
On the other hand, for any partition $\mu$ and matrices $A,B\in\Mat_{N,N}(\mbC)$ we have~\cite[eq.~(31)]{Ale11}
$$
\int_{\mathcal U_N}s_\lambda(p_*(UA))s_\mu(p_*(\oU^tB))dU=\frac{s_\lambda(p_*(AB))}{s_\lambda(p_*(\id_N))}\delta_{\lambda,\mu}.
$$
Therefore, we obtain
\begin{align*}
&\int_{{\mathcal U}_N}s_\lambda(p_*(U))e^{\tr\oU^t A}dU=\int_{{\mathcal U}_N}s_\lambda(p_*(U))\sum_{k\ge 0}\frac{p_1(\oU^t A)^k}{k!}dU=\\
=&\int_{{\mathcal U}_N}s_\lambda(p_*(U))\sum_\mu s_\mu(1,0,0,\ldots)s_\mu(p_*(\oU^tA))dU=\frac{1}{C_N(\lambda)}s_\lambda(p_*(A)).
\end{align*}
This completes the proof of the lemma.
\end{proof}

In particular, we have
$$
\left.\ttau^o\right|_{q_i=q_i(A)}=\int_{{\mathcal U}_N} \left.\tau^{o,ext}_N\right|_{s_i=i! \tr U^{i+1}}e^{\tr {\overline U}^tA}dU.
$$

Equations~\eqref{eq:refined matrix model},~\eqref{eq:extended matrix},~\eqref{eq:full refined matrix model} and~\eqref{eq:note} imply that
$$
\left.\tau^{o,ext}_N\right|_{s_{\ge 1}=0}=\tau^o_N,\qquad \tau^{o,ext}_1=\tau^{o,ext}.
$$
We conjecture that there exists a geometric construction of boundary descendents in the refined open intersection theory giving the extended refined open partition function $\tau^{o,ext}_N$.

The function
$$
F^{o,ext,N}:=\log\tau^{o,ext}_N-F^c
$$
will be called the \emph{extended refined open free energy}.

\subsection{Feynman diagram expansion of the extended matrix model}\label{subsection:Feynman for extended}

Introduce the \emph{extended refined open intersection numbers} by
\begin{gather}\label{eq:extended refined numbers}
\<\tau_{a_1}\cdots\tau_{a_l}\sigma_{c_1}\cdots\sigma_{c_k}\>^{o,ext,N}:=\left.\frac{\d^{l+k}F^{o,ext,N}}{\d t_{a_1}\cdots\d t_{a_l}\d s_{c_1}\cdots\d s_{c_k}}\right|_{t_*=s_*=0}.
\end{gather}
From~\eqref{eq:transformation} it follows that the intersection number $\<\tau_{a_1}\cdots\tau_{a_l}\sigma_{c_1}\cdots\sigma_{c_k}\>^{o,ext,N}$ is actually a polynomial in $N$ with rational coefficients. So, it is well-defined for all values of~$N$, not necessarily positive integers. Therefore, the extended refined open partition function~$\tau^{o,ext}_N$ is also well-defined for all values of~$N$. We want to write a combinatorial formula for the extended refined open intersection numbers similar to~\eqref{eq:refined combinatorial formula}. Let us write the matrix model~\eqref{eq:full refined matrix model} for~$\tau^{o,ext}_N$ in the following way:
\begin{align}\label{eq:full refined matrix model,2}
&\left.\tau^{o,ext}_N\right|_{t_i=t_i(\Lambda)}=\frac{c_{\Lambda,M}}{(2\pi)^{N^2}}\int_{\mcH_M\times\Mat_{N,N}(\mbC)}e^{-\frac{1}{2}\tr H^2\Lambda-\frac{1}{2}\tr Z\oZ^t}e^{\frac{1}{6}\tr H^3+\frac{1}{6}\tr Z^3}\times\\
&\times\det\frac{\Lambda\otimes\id_N+\sqrt{\Lambda^2\otimes\id_N-\id_M\otimes\oZ^t}-H\otimes\id_N+\id_M\otimes Z}{\Lambda\otimes\id_N+\sqrt{\Lambda^2\otimes\id_N-\id_M\otimes\oZ^t}-H\otimes\id_N-\id_M\otimes Z}e^{\sum_{i\ge 0}\frac{2^{-i-1}}{(i+1)!}s_i\tr(\oZ^t)^{i+1}}dHdZ.\notag
\end{align}
We see that this matrix model is obtained from~\eqref{eq:refined matrix model} simply by adding the factor $e^{\sum_{i\ge 1}\frac{2^{-i-1}}{(i+1)!}s_i\tr(\oZ^t)^{i+1}}$ in the integrand. Doing the Feynman diagram expansion of~\eqref{eq:full refined matrix model,2} one can easily see that there is a combinatorial formula for the intersection numbers~\eqref{eq:extended refined numbers} similar to~\eqref{eq:refined combinatorial formula}, where we allow odd critical nodal ribbon graphs with boundary to have certain exceptional components. Let us formulate it precisely.

Recall that a $(g,k,l)$-ribbon graph with boundary is called critical, if
\begin{itemize}
\item Boundary marked points have degree $2$.
\item All other vertices have degree $3$.
\item If $l=0$, then $g=0$ and $k=3$.
\end{itemize}
We will call a $(g,k,l)$-ribbon graph with boundary \emph{exceptional}, if $g=l=0$ and $k\ge 1$. Obviously, for each $k\ge 1$ there exists a unique such graph up to an isomorphism, see Fig.~\ref{fig:exceptional component}.
\begin{figure}[t]
\begin{tikzpicture}[scale=0.7]

\draw [pattern=grid, pattern color=black] (0,0) circle (1.5);
\coordinate [label=center:\textbullet] (B) at (1.5/1.41,1.5/1.41);
\coordinate [label=center:\textbullet] (B) at (-1.5/1.41,1.5/1.41);
\coordinate [label=center:\textbullet] (B) at (-1.5/1.41,-1.5/1.41);
\coordinate [label=center:\textbullet] (B) at (1.5/1.41,-1.5/1.41);

\begin{scope}[shift={(9,0)}]
\draw [pattern=grid, pattern color=black] (0,0) circle (1.5);

\draw (1.05+0.16,-1.05+0.16) -- (1.05+1.5/1.41+0.16,-1.05-1.5/1.41+0.16);
\draw (1.05-0.16,-1.05-0.16) -- (1.05+1.5/1.41-0.16,-1.05-1.5/1.41-0.16);
\fill [pattern=dots, pattern color=black] (1.05+0.16,-1.05+0.16) -- (1.05+1.5/1.41+0.16,-1.05-1.5/1.41+0.16) -- (1.05+1.5/1.41-0.16,-1.05-1.5/1.41-0.16) --  (1.05-0.16,-1.05-0.16) -- cycle;
\draw [thick,dotted] (1.05+1.5/1.41+0.16,-1.05-1.5/1.41+0.16) -- (1.05+1.5/1.41-0.16,-1.05-1.5/1.41-0.16);
\draw (1.05+0.16,1.05-0.16) -- (1.05+1.5/1.41+0.16,1.05+1.5/1.41-0.16);
\draw (1.05-0.16,1.05+0.16) -- (1.05+1.5/1.41-0.16,1.05+1.5/1.41+0.16);
\fill [pattern=dots, pattern color=black] (1.05+0.16,1.05-0.16) -- (1.05+1.5/1.41+0.16,1.05+1.5/1.41-0.16) -- (1.05+1.5/1.41-0.16,1.05+1.5/1.41+0.16) --  (1.05-0.16,1.05+0.16) -- cycle;
\draw [thick,dotted] (1.05+1.5/1.41+0.16,1.05+1.5/1.41-0.16) -- (1.05+1.5/1.41-0.16,1.05+1.5/1.41+0.16);
\draw (-1.05-0.16,1.05-0.16) -- (-1.05-1.5/1.41-0.16,1.05+1.5/1.41-0.16);
\draw (-1.05+0.16,1.05+0.16) -- (-1.05-1.5/1.41+0.16,1.05+1.5/1.41+0.16);
\fill [pattern=dots, pattern color=black] (-1.05-0.16,1.05-0.16) -- (-1.05-1.5/1.41-0.16,1.05+1.5/1.41-0.16) -- (-1.05-1.5/1.41+0.16,1.05+1.5/1.41+0.16) --  (-1.05+0.16,1.05+0.16) -- cycle;
\draw [thick,dotted] (-1.05-1.5/1.41-0.16,1.05+1.5/1.41-0.16) -- (-1.05-1.5/1.41+0.16,1.05+1.5/1.41+0.16);
\draw (-1.05-0.16,-1.05+0.16) -- (-1.05-1.5/1.41-0.16,-1.05-1.5/1.41+0.16);
\draw (-1.05+0.16,-1.05-0.16) -- (-1.05-1.5/1.41+0.16,-1.05-1.5/1.41-0.16);
\fill [pattern=dots, pattern color=black] (-1.05-0.16,-1.05+0.16) -- (-1.05-1.5/1.41-0.16,-1.05-1.5/1.41+0.16) -- (-1.05-1.5/1.41+0.16,-1.05-1.5/1.41-0.16) --  (-1.05+0.16,-1.05-0.16) -- cycle;
\draw [thick,dotted] (-1.05-1.5/1.41-0.16,-1.05-1.5/1.41+0.16) -- (-1.05-1.5/1.41+0.16,-1.05-1.5/1.41-0.16);

\coordinate [label=center:$-$] (B) at (1.65,0.85);
\coordinate [label=center:$-$] (B) at (0.9,-1.7);
\coordinate [label=center:$-$] (B) at (-1.65,-0.85);
\coordinate [label=center:$-$] (B) at (-0.9,1.7);

\begin{scope}[shift={(0.6,0)}]
\draw [thick] (2.9,0) -- (4.4,0);
\draw [thick] (2.9,0.15) -- (2.9,-0.15);
\fill (4.15,0.15)--(4.4,0)--(4.15,-0.15)--cycle;
\coordinate [label=0: $\frac{s_3}{2^4 4!}\mathrm{tr}(\overline{Z}^t)^4$] (B) at (5.5,0);
\end{scope}

\end{scope}

\end{tikzpicture}
\caption{Exceptional graph and the corresponding diagram for the extended matrix model}
\label{fig:exceptional component}
\end{figure}
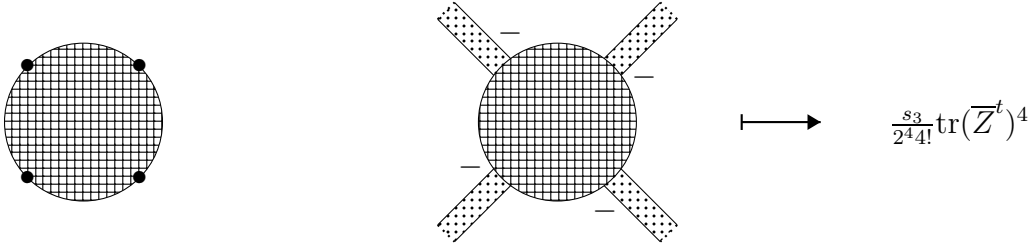
Note, that we step back a little bit from the original definition of a ribbon graph with boundary, because exceptional graphs with $k=1$ or $k=2$ are strictly speaking not stable. Note also that a critical $(0,3,0)$-ribbon graph with boundary, that we call a ghost, coincides with an exceptional graph with $k=3$. However, our idea is to distinguish them. Speaking formally, to a $(0,3,0)$-ribbon graph with boundary we additionally assign a type: it can be a ghost or an exceptional graph. Using this terminology, the set of critical ribbon graphs with boundary does not intersect the set of exceptional graphs.

A nodal ribbon graph with boundary~$G=\left(\coprod_i G_i\right)/N$ will be called \emph{extended critical}, if
\begin{itemize}
\item It does not have boundary marked points.
\item All of its components~$G_i$ are critical or exceptional.
\item Ghost components do not contain the illegal sides of nodes.
\item Exceptional components do not contain the legal sides of nodes.
\end{itemize}
The fact, that we do not allow boundary marked points now, may look surprising, but one can note that an exceptional component with $k=1$ can be easily interpreted as a boundary marked point. An extended critical nodal ribbon graph with boundary~$G=\left(\coprod_i G_i\right)/N$ is called odd if any boundary component of each non-exceptional~$G_i$ has an odd number of the legal sides of nodes. Denote by $\tcR^{ext}_{l}$ the set of odd extended critical nodal ribbon graphs with boundary with~$l$ internal faces. For a graph $G=\left(\coprod_i G_i\right)/N\in\tcR^{ext}_{l}$ introduce the following notations. Denote by $b(G)$ the number of boundary components in a smoothing of the nodal surface associated with~$G$. Let $c(G):=\prod_i c(G_i)$, where $c(G_i)$ is defined by~\eqref{eq:definition of constant} if $G_i$ is non-exceptional and
$$
c(G_i):=\frac{1}{m!}, \quad\text{if $G_i$ is an exceptional graph with $m+1$ boundary vertices, $m\ge 0$}.
$$
For $m\ge 0$ denote by $\exc_m(G)$ the number of exceptional components $G_i$ with exactly $m+1$ boundary vertices. The set of edges $\Edges(G)$ is composed of the internal edges of the $G_i$'s and of the boundary edges. The boundary edges are the boundary segments in non-exceptional~$G_i$'s between successive legal sides of nodes. For an edge $e\in\Edges(G)$ the function $\lambda(e)$ is defined by the old formula~\eqref{eq:definition of lambda}. The Feynman diagram expansion of the matrix model~\eqref{eq:full refined matrix model,2} gives the following formula for the intersection numbers~\eqref{eq:extended refined numbers}:
\begin{multline}\label{eq:full combinatorial formula}
\sum_{a_1,\ldots,a_l\ge 0}\sum_{m\ge 0}\sum_{c_1,\ldots,c_m\ge 0}\<\tau_{a_1}\cdots\tau_{a_l}\sigma_{c_1}\cdots\sigma_{c_m}\>^{o,ext,N}\prod_{i=1}^l\frac{(2a_i-1)!!}{\lambda_i^{2a_i+1}}\frac{\prod_{j=1}^m s_{c_j}}{m!}=\\
=\sum_{G=\left(\coprod_i G_i\right)/N\in\tcR^{ext}_l}\frac{c(G)}{|\Aut(G)|}N^{b(G)}\prod_{e\in\Edges(G)}\lambda(e)\prod_{m\ge 0}s_m^{\exc_m(G)}.
\end{multline}

\subsection{String equation}\label{subsection:string equation}

\begin{proposition}\label{proposition:string equation}
We have the string equation
\begin{gather}\label{eq:string equation}
\left(\frac{\d}{\d t_0}-\sum_{i\ge 0} t_{i+1}\frac{\d}{\d t_i}-\sum_{i\ge 0}s_{i+1}\frac{\d}{\d s_i}-\frac{t_0^2}{2}-N s_0\right)\tau^{o,ext}_N=0.
\end{gather}
\end{proposition}
\begin{proof}
We will use formula~\eqref{eq:full refined matrix model,2}. Denote
\begin{align*}
I_1:=&e^{\frac{1}{6}\tr H^3+\frac{1}{6}\tr Z^3-\frac{1}{2}\tr H^2\Lambda-\frac{1}{2}\tr Z\oZ^t},\\
I_2:=&\det\frac{\Lambda\otimes\id_N+\sqrt{\Lambda^2\otimes\id_N-\id_M\otimes\oZ^t}-H\otimes\id_N+\id_M\otimes Z}{\Lambda\otimes\id_N+\sqrt{\Lambda^2\otimes\id_N-\id_M\otimes\oZ^t}-H\otimes\id_N-\id_M\otimes Z},\\
I_3:=&e^{\sum_{i\ge 0}\frac{2^{-i-1}}{(i+1)!}s_i\tr(\oZ^t)^{i+1}},\\
Z_{M,N}:=&\left.\tau^{o,ext}_N\right|_{t_i=t_i(\Lambda)}=\frac{c_{\Lambda,M}}{(2\pi)^{N^2}}\int_{\mcH_M\times\Mat_{N,N}(\mbC)}I_1I_2I_3 dHdZ.
\end{align*}
Our approach is a modification of the diagrammatic method of E.~Witten (\cite{Wit92}) that he used for a proof of the Virasoro equations for the closed partition function~$\tau^c$. First of all, note that
\begin{align}
&\left.\left(-\sum_{i\ge 0} t_{i+1}\frac{\d}{\d t_i}\right)\tau^{o,ext}_N\right|_{t_j=t_j(\Lambda)}=\sum_{i=1}^M\frac{1}{\lambda_i}\frac{\d}{\d\lambda_i}Z_{M,N}=\label{comp1}\\
=&\frac{t_0(\Lambda)^2}{2}Z_{M,N}+\frac{c_{\Lambda,M}}{(2\pi)^{N^2}}\int_{\mcH_M\times\Mat_{N,N}(\mbC)}\left(-\frac{I_2}{2}\tr H^2\Lambda^{-1}+\sum_{i=1}^M\frac{1}{\lambda_i}\frac{\d I_2}{\d\lambda_i}\right)I_1I_3dHdZ\notag
\end{align}
and
\begin{align}
&\left.\left(-\sum_{i\ge 0}s_{i+1}\frac{\d}{\d s_i}-Ns_0\right)\tau^{o,ext}_N\right|_{t_j=t_j(\Lambda)}=\left(-\sum_{i\ge 0}s_{i+1}\frac{\d}{\d s_i}-N s_0\right)Z_{M,N}=\label{comp2}\\
=&\frac{c_{\Lambda,M}}{(2\pi)^{N^2}}\int_{\mcH_M\times\Mat_{N,N}(\mbC)}\left(-\sum_{i\ge 0}\frac{2^{-i}}{i!}s_i\tr(\oZ^t)^i\right)I_1I_2I_3 dHdZ=\notag\\
=&\frac{c_{\Lambda,M}}{(2\pi)^{N^2}}\int_{\mcH_M\times\Mat_{N,N}(\mbC)}I_1I_2(-2)\sum_{i=1}^N\frac{\d I_3}{\d\oz_{i,i}}dHdZ.\notag
\end{align}

The only non-trivial step in the proof is to express the derivative $\left.\frac{\d\tau^{o,ext}_N}{\d t_0}\right|_{t_i=t_i(\Lambda)}$, as a matrix integral. Let us prove that
\begin{gather}\label{eq:t0 derivative}
\left.\frac{\d\tau^{o,ext}_N}{\d t_0}\right|_{t_i=t_i(\Lambda)}=\frac{c_{\Lambda,M}}{(2\pi)^{N^2}}\int_{\mcH_M\times\Mat_{N,N}(\mbC)}\left(\tr H+\tr Z\right)I_1I_2I_3dHdZ.
\end{gather}
The $t_0$ derivative corresponds to an extra insertion of $\tau_0$ on the left-hand side of~\eqref{eq:extended refined numbers}. We want to consider the generating function from the left-hand side of~\eqref{eq:full combinatorial formula} with an extra insertion of $\tau_0$. In order to get it from the right-hand side of~\eqref{eq:full combinatorial formula}, we have to sum over graphs $G=\left(\coprod_iG_i\right)/N\in\tcR^{ext}_{l+1}$ with a distinguished face, which we call $C_0$, labeled with a variable~$\lambda_0$, then consider the behavior for $\lambda_0\to\infty$ and extract the coefficient of $\frac{1}{\lambda_0}$. The coefficient of $\frac{1}{\lambda_0}$ comes precisely from graphs, where the face $C_0$ has only one edge. The structure of the neighborhood of the distinguished face in such graphs is indicated in Fig.~\ref{fig:pic for string}.
\begin{figure}[t]
\begin{tikzpicture}[scale=0.7]

\draw (0,0) circle (1.8);
\draw (0,0) circle (1.2);
\fill [color=white] (-1.7,0.3) -- (-1.9,0.3) -- (-1.9,-0.3) --  (-1.7,-0.3) -- cycle;

\draw (-1.8+0.03,0.3) -- (-3,0.3);
\draw (-1.8+0.03,-0.3) -- (-3,-0.3);

\draw [thick,dotted] (-3,0.3) -- (-3,-0.3);

\draw [dashed] (-2.3,2.3) -- (2.3,2.3) -- (2.3,-2.3) --  (-2.3,-2.3) -- cycle;
\coordinate [label=center: $C_0$] (B) at (0,0);

\begin{scope}[shift={(8,0)}]

\draw (0,0) circle (1.8);

\draw [dashed] (-1.8,2.3) -- (2.3,2.3) -- (2.3,-2.3) --  (-1.8,-2.3) -- cycle;
\coordinate [label=center: $C_0$] (B) at (0,0);

\draw (-1.8,0) arc (0:36.87:2.5);
\draw (-1.8,0) arc (0:-36.87:2.5);
\draw [thick,dotted] (-2.3,1.5) arc (36.87:57:2.5);
\draw [thick,dotted] (-2.3,-1.5) arc (-36.87:-57:2.5);

\coordinate [label=center: $+$] (B) at (-1.5,0);
\coordinate [label=center: $-$] (B) at (-2.1,0);

\end{scope}

\end{tikzpicture}
\caption{Graphs that dominate for $\lambda_0\to\infty$}
\label{fig:pic for string}
\end{figure}
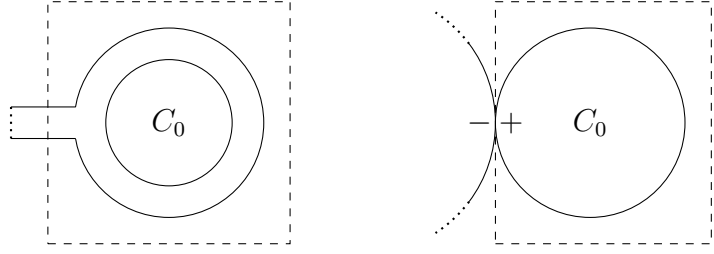
We see that there are two cases. In the first case, the edge of our face is internal. In the second case, the edge of the face is boundary. Then, automatically, the face belongs to a component~$G_i$ of type~$(0,1,1)$. The first picture in Fig.~\ref{fig:pic for string} already appeared in~\cite{Wit92} in the diagrammatic proof of the string equation for~$\tau^c$. The contribution of this picture in our situation is computed in exactly the same way, as in~\cite{Wit92}, and it gives the first term~$\tr H$ in the brackets on the right-hand side of~\eqref{eq:t0 derivative}. Consider the second picture in Fig.~\ref{fig:pic for string}. A graph outside the dotted lines can be an arbitrary odd extended critical nodal ribbon graph with an additional distinguished illegal "half" of a node. The part inside the dotted lines gives $\frac{1}{\lambda_0}$. So, in order to get the contribution of the second picture, we should sum over all exteriors. It is easy to see that this sum gives the second term~$\tr Z$ in the brackets on the right-hand side of~\eqref{eq:t0 derivative}.

Computations~\eqref{comp1},~\eqref{comp2} and~\eqref{eq:t0 derivative} show that the string equation~\eqref{eq:string equation} is equivalent to the equation
\begin{gather}\label{eq:string equation2}
\int_{\mcH_M\times\Mat_{N,N}(\mbC)}\left[\left(-\frac{\tr H^2\Lambda^{-1}}{2}+\tr H+\tr Z\right)I_2I_3+I_3\sum_{i=1}^M\frac{1}{\lambda_i}\frac{\d I_2}{\d \lambda_i}-2I_2\sum_{i=1}^N\frac{\d I_3}{\d\oz_{i,i}}\right]I_1dHdZ=0.
\end{gather}
Note that
$$
\left[\sum_{i=1}^M\frac{1}{\lambda_i}\left(\frac{\d}{\d\lambda_i}+\frac{\d}{\d h_{i,i}}\right)+2\sum_{j=1}^N\frac{\d}{\d\oz_{j,j}}\right]I_2=0.
$$
Therefore, equation~\eqref{eq:string equation2} is equivalent to
\begin{gather}\label{eq:string equation3}
\int_{\mcH_M\times\Mat_{N,N}(\mbC)}\left[\left(-\frac{\tr H^2\Lambda^{-1}}{2}+\tr H+\tr Z\right)I_2I_3-I_3\sum_{i=1}^M\frac{1}{\lambda_i}\frac{\d I_2}{\d h_{i,i}}-2\sum_{i=1}^N\frac{\d(I_2I_3)}{\d\oz_{i,i}}\right]I_1dHdZ=0.
\end{gather}
Applying the relations
\begin{align*}
0=&\int_{\mcH_M\times\Mat_{N,N}(\mbC)}\sum_{i=1}^M\frac{1}{\lambda_i}\frac{\d(I_1I_2I_3)}{\d h_{i,i}}dHdZ=\\
=&\int_{\mcH_M\times\Mat_{N,N}(\mbC)}\left[\left(\frac{\tr H^2\Lambda^{-1}}{2}-\tr H\right)I_2I_3+I_3\sum_{i=1}^M\frac{1}{\lambda_i}\frac{\d I_2}{\d h_{i,i}}\right]I_1dHdZ,\\
0=&\int_{\mcH_M\times\Mat_{N,N}(\mbC)}\sum_{j=1}^N\frac{\d(I_1I_2I_3)}{\d\oz_{j,j}}dHdZ=\\
=&\int_{\mcH_M\times\Mat_{N,N}(\mbC)}\left[-\frac{\tr Z}{2}I_2I_3+\sum_{i=1}^N\frac{\d(I_2I_3)}{\d\oz_{i,i}}\right]I_1dHdZ,
\end{align*}
we see that equation~\eqref{eq:string equation3} is true. The string equation~\eqref{eq:string equation} is proved.
\end{proof}

\subsection{Dilaton equation}\label{subsection:dilaton equation}

\begin{proposition}\label{proposition:dilaton equation}
We have the dilaton equation
\begin{gather}\label{eq:dilaton equation}
\left(\frac{\d}{\d t_1}-\sum_{n\ge 0}\frac{2n+1}{3}t_n\frac{\d}{\d t_n}-\sum_{n\ge 0}\frac{2n+2}{3}s_n\frac{\d}{\d s_n}-\frac{N^2}{2}-\frac{1}{24}\right)\tau^{o,ext}_N=0.
\end{gather}
\end{proposition}
\begin{proof}
We have
\begin{multline}\label{eq:dilaton,formula1}
\left.\left(-\sum_{i\ge 0}\frac{2i+1}{3}t_i\frac{\d}{\d t_i}\right)\tau^{o,ext}_N\right|_{t_j=t_j(\Lambda)}=\frac{1}{3}\sum\lambda_i\frac{\d}{\d\lambda_i}Z_{M,N}=\\
=\frac{c_{\Lambda,M}}{(2\pi)^{N^2}}\int_{\mcH_M\times\Mat_{N,N}(\mbC)}\left[\left(\frac{M^2}{6}-\frac{\tr H^2\Lambda}{6}\right)I_1I_2I_3+\frac{1}{3}I_1I_3\sum\lambda_i\frac{\d I_2}{\d\lambda_i}\right]dHdZ.
\end{multline}
It is also easy to see that
\begin{gather}\label{eq:dilaton,formula2}
\left.\left(-\sum_{i\ge 0}\frac{2i+2}{3}s_i\frac{\d}{\d s_i}\right)\tau^{o,ext}_N\right|_{t_j=t_j(\Lambda)}=\frac{c_{\Lambda,M}}{(2\pi)^{N^2}}\int_{\mcH_M\times\Mat_{N,N}(\mbC)}\left[-\frac{2}{3}I_1I_2\sum\oz_{i,j}\frac{\d I_3}{\d\oz_{i,j}}\right]dHdZ.
\end{gather}

As in the proof of the string equation, the only non-trivial step here is the computation of the $t_1$ derivative. Let us prove that
\begin{align}
\left.\frac{\d\tau^{o,ext}_N}{\d t_1}\right|_{t_i=t_i(\Lambda)}=&\frac{c_{\Lambda,M}}{(2\pi)^{N^2}}\int_{\mcH_M\times\Mat_{N,N}(\mbC)}\left[\left(\frac{1}{3}\tr H^3-\tr H^2\Lambda+\tr H\Lambda^2+\frac{M^2}{2}+\frac{1}{24}\right)I_1I_2I_3\right.\label{eq:t1 derivative,first line}\\
&\hspace{4cm}+I_1I_3\left(-\sum\lambda_i\frac{\d I_2}{\d h_{i,i}}+\frac{1}{2}\sum z_{i,j}\frac{\d I_2}{\d z_{i,j}}+\sum h_{i,j}\frac{\d I_2}{\d h_{i,j}}\right)\label{eq:t1 derivative,second line}\\
&\hspace{4cm}\left.+\left(\frac{1}{12}\tr Z^3+\frac{1}{4}\tr Z\oZ^t\right)I_1I_2I_3\right]dHdZ.\label{eq:t1 derivative,third line}
\end{align}
We want to compute the generating series from the left-hand side of~\eqref{eq:full combinatorial formula} with an extra insertion of~$\tau_1$. In order to get it from the right-hand side of~\eqref{eq:full combinatorial formula}, we have to sum over graphs $G=\left(\coprod_iG_i\right)/N\in\tcR^{ext}_{l+1}$ with a distinguished face, which we call $C_0$, labeled with a variable~$\lambda_0$, and then pick out the coefficient of $\frac{1}{\lambda_0^3}$. The coefficient of $\frac{1}{\lambda_0^3}$ can only come from graphs, where the face~$C_0$ has at most three edges. The structure of such graphs in indicated in Fig.~\ref{fig:pic for dilaton} and Fig.~\ref{fig:pic for dilaton2}.
\begin{figure}[t]
\begin{tikzpicture}[scale=0.7]

\draw (0,0) circle (1.8);
\draw (0,0) circle (1.2);
\fill [color=white] (-1.7,0.3) -- (-1.9,0.3) -- (-1.9,-0.3) --  (-1.7,-0.3) -- cycle;

\draw (-1.8+0.03,0.3) -- (-3,0.3);
\draw (-1.8+0.03,-0.3) -- (-3,-0.3);

\draw [thick,dotted] (-3,0.3) -- (-3,-0.3);

\draw [dashed] (-2.3,2.3) -- (2.3,2.3) -- (2.3,-2.3) --  (-2.3,-2.3) -- cycle;
\coordinate [label=center: $C_0$] (B) at (0,0);

\begin{scope}[shift={(7,0)}]

\draw (0,0) circle (1.8);
\draw (0,0) circle (1.2);
\fill [color=white] (-1.7,0.3) -- (-1.9,0.3) -- (-1.9,-0.3) --  (-1.7,-0.3) -- cycle;

\draw (-1.8+0.03,0.3) -- (-3,0.3);
\draw (-1.8+0.03,-0.3) -- (-3,-0.3);

\draw [thick,dotted] (-3,0.3) -- (-3,-0.3);

\fill [color=white] (1.7,0.3) -- (1.9,0.3) -- (1.9,-0.3) --  (1.7,-0.3) -- cycle;

\draw (1.8-0.03,0.3) -- (3,0.3);
\draw (1.8-0.03,-0.3) -- (3,-0.3);

\draw [thick,dotted] (3,0.3) -- (3,-0.3);

\draw [dashed] (-2.3,2.3) -- (2.3,2.3) -- (2.3,-2.3) --  (-2.3,-2.3) -- cycle;
\coordinate [label=center: $C_0$] (B) at (0,0);

\end{scope}

\begin{scope}[shift={(13.8,-0.3)}]

\draw (0.3,2.4-0.52) -- (0.87*2.4-0.3,-1.2+0.52);
\draw (-0.3,2.4-0.52) -- (-0.87*2.4+0.3,-1.2+0.52);
\draw (0.87*2.4-0.6,-1.2) -- (-0.87*2.4+0.6,-1.2);

\draw (0.3,2.4-0.52) -- (0.3,2.4-0.52+1.2);
\draw (-0.3,2.4-0.52) -- (-0.3,2.4-0.52+1.2);
\draw [thick,dotted] (0.3,2.4-0.52+1.2) -- (-0.3,2.4-0.52+1.2);

\draw (0.87*2.4-0.3,-1.2+0.52) -- (0.87*2.4-0.3+1.04,-1.2+0.52-0.6);
\draw (0.87*2.4-0.6,-1.2) -- (0.87*2.4-0.6+1.04,-1.2-0.6);
\draw [thick,dotted] (0.87*2.4-0.3+1.04,-1.2+0.52-0.6) -- (0.87*2.4-0.6+1.04,-1.2-0.6);

\draw (-0.87*2.4+0.3,-1.2+0.52) -- (-0.87*2.4+0.3-1.04,-1.2+0.52-0.6);
\draw (-0.87*2.4+0.6,-1.2) -- (-0.87*2.4+0.6-1.04,-1.2-0.6);
\draw [thick,dotted] (-0.87*2.4+0.3-1.04,-1.2+0.52-0.6) -- (-0.87*2.4+0.6-1.04,-1.2-0.6);

\draw (0,1.2) -- (0.87*1.2,-0.6);
\draw (0,1.2) -- (-0.87*1.2,-0.6);
\draw (0.87*1.2,-0.6) -- (-0.87*1.2,-0.6);

\draw [dashed] (-2.1,2.4) -- (2.1,2.4) -- (2.1,-1.9) --  (-2.1,-1.9) -- cycle;
\coordinate [label=center: $C_0$] (B) at (0,0);

\end{scope}

\begin{scope}[shift={(3.5,-5.5)}]

\draw (-1.6,0) circle (1.2);
\draw (-1.6,0) circle (0.6);
\draw (1.6,0) circle (1.2);
\draw (1.6,0) circle (0.6);

\fill [color=white] (-0.6,0.3) -- (-0.6,-0.3) -- (0.6,-0.3) --  (0.6,0.3) -- cycle;
\draw (-0.45,0.3) -- (0.45,0.3);
\draw (-0.45,-0.3) -- (0.45,-0.3);

\coordinate [label=center: $C_0$] (B) at (0,0.8);

\end{scope}

\begin{scope}[shift={(10.5,-5.5)}]

\draw (-1.5,0.3) -- (-0.3,0.3);
\draw (-1.5,-0.3) -- (-0.3,-0.3);
\draw (0.3,0.3) -- (0.3,-0.3);
\draw (-2.1,0.3) -- (-2.1,-0.3);

\draw (-1.5,0.3) .. controls (-1.5,1.5) and (0.75,1.8) .. (1.25,1.05);
\draw (1.25,1.05) .. controls (2,-0.075) and (1,-1.1) .. (0.3,-0.3);
\draw (-2.1,0.3) .. controls (-2.1,1.7) and (0.5,2.5) .. (1.4,1.7);
\draw (1.4,1.7) .. controls (3.2,0.1) and (1,-2.4) .. (-0.3,-0.3);

\draw (-1.5,-0.3) .. controls (-1.5,-1.5) and (0.75,-1.8) .. (1.25,-1.05);
\draw (1.25,-1.05) .. controls (2,0.075) and (1,1.1) .. (0.3,0.3);
\draw (-2.1,-0.3) .. controls (-2.1,-1.7) and (0.5,-2.5) .. (1.4,-1.7);
\draw (1.4,-1.7) .. controls (3.2,-0.1) and (1,2.4) .. (-0.3,0.3);

\fill [color=white] (1.447,1.131) circle (0.2);
\fill [color=white] (1.845,0.06) circle (0.2);
\fill [color=white] (1.845,0.25) circle (0.2);
\fill [color=white] (1.8,0.58) circle (0.2);
\fill [color=white] (1.6,0.9) circle (0.2);
\fill [color=white] (1.687,-0.072) circle (0.2);
\fill [color=white] (1.189,-0.825) circle (0.2);
\fill [color=white] (1.38,-0.65) circle (0.2);
\fill [color=white] (1.6,-0.4) circle (0.2);
\fill [color=white] (1.65,-0.25) circle (0.2);

\coordinate [label=center: $C_0$] (B) at (-1.9,1.7);

\end{scope}

\end{tikzpicture}
\caption{Graphs of internal type that contribute in order $\frac{1}{\lambda_0^3}$}
\label{fig:pic for dilaton}
\end{figure}
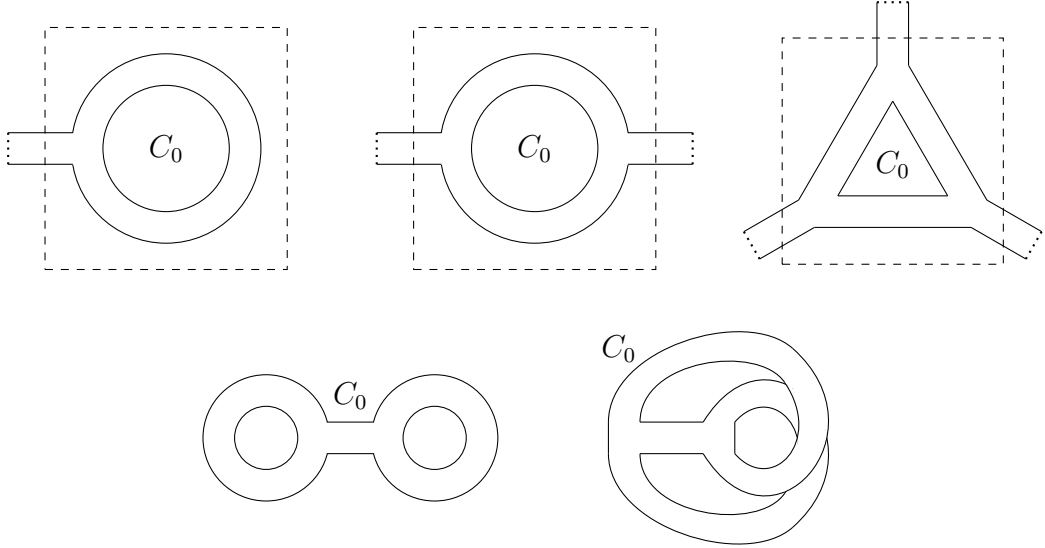
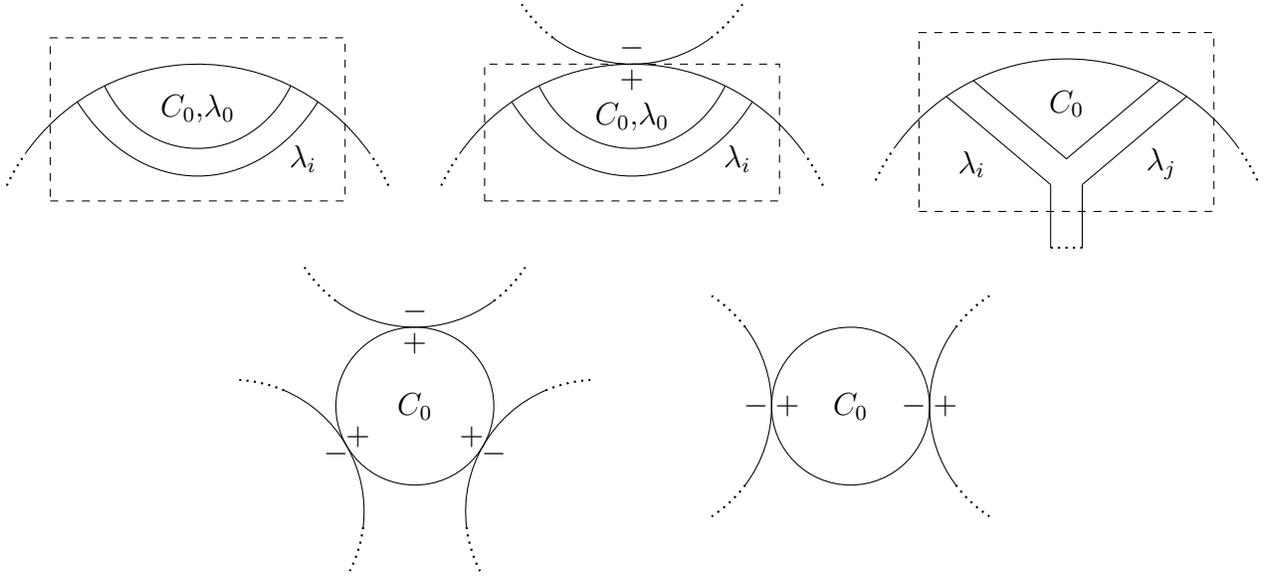
\begin{figure}[t]
\begin{tikzpicture}[scale=0.7]

\draw (1.77,3.59) .. controls (1,2) and (-1,2) .. (-1.77,3.59);
\draw (1.75+0.6*3.59/4,3.555-0.6*1.77/4) .. controls (1.1,1.4) and (-1.1,1.4) .. (-1.75-0.6*3.59/4,3.555-0.6*1.77/4);

\draw (3.277,2.294) arc (35:145:4);
\draw [thick,dotted] (3.625,1.690) arc (25:35:4);
\draw [thick,dotted] (-3.625,1.690) arc (155:145:4);
\coordinate [label=center: $C_0{,} \lambda_0$] (B) at (0,3.15);
\coordinate [label=center: $\lambda_i$] (B) at (2,2.2);
\draw[dashed] (-2.8,4.5) -- (2.8,4.5) -- (2.8,1.4) -- (-2.8,1.4) -- cycle;

\begin{scope}[shift={(8.25,0)}]

\draw (1.77,3.59) .. controls (1,2) and (-1,2) .. (-1.77,3.59);
\draw (1.75+0.6*3.59/4,3.555-0.6*1.77/4) .. controls (1.1,1.4) and (-1.1,1.4) .. (-1.75-0.6*3.59/4,3.555-0.6*1.77/4);

\draw (3.277,2.294) arc (35:145:4);
\draw [thick,dotted] (3.625,1.690) arc (25:35:4);
\draw [thick,dotted] (-3.625,1.690) arc (155:145:4);
\coordinate [label=center: $C_0{,} \lambda_0$] (B) at (0,3);
\coordinate [label=center: $+$] (B) at (0,3.7);
\coordinate [label=center: $-$] (B) at (0,4.3);
\coordinate [label=center: $\lambda_i$] (B) at (2,2.2);
\draw[thick,dotted] (1.5,4.5) arc (36.87-90:57-90:2.5);
\draw[thick,dotted] (-1.5,4.5) arc (-36.87-90:-57-90:2.5);
\draw (1.5,4.5) arc (36.87-90:-36.87-90:2.5);

\draw[dashed] (-2.8,4) -- (2.8,4) -- (2.8,1.4) -- (-2.8,1.4) -- cycle;

\end{scope}

\begin{scope}[shift={(16.5,0.1)}]

\draw (1.77,3.59) -- (0,2.1) -- (-1.77,3.59);
\draw (1.75+0.6*3.59/4,3.555-0.6*1.77/4) -- (0.3,1.36+0.25);
\draw (-0.3,1.36+0.25) -- (-1.75-0.6*3.59/4,3.555-0.6*1.77/4);
\draw (-0.3,1.36+0.25) -- (-0.3,1.36+0.25-1.2);
\draw (0.3,1.36+0.25) -- (0.3,1.36+0.25-1.2);
\draw [thick,dotted] (-0.3,1.36+0.25-1.2) -- (0.3,1.36+0.25-1.2);
\draw (3.277,2.294) arc (35:145:4);
\draw [thick,dotted] (3.625,1.690) arc (25:35:4);
\draw [thick,dotted] (-3.625,1.690) arc (155:145:4);
\coordinate [label=center: $C_0$] (B) at (0,3.15);
\coordinate [label=center: $\lambda_j$] (B) at (1.8,2);
\coordinate [label=center: $\lambda_i$] (B) at (-1.8,2);

\draw[dashed] (-2.8,4.5) -- (2.8,4.5) -- (2.8,1.1) -- (-2.8,1.1) -- cycle;

\end{scope}

\begin{scope}[shift={(4.125,-2.5)}]
\draw (0,0) circle (1.5);
\coordinate [label=center:$+$] (B) at (0,1.2);
\coordinate [label=center:$-$] (B) at (0,1.8);
\coordinate [label=center:$+$] (B) at (-1.08,-0.58);
\coordinate [label=center:$-$] (B) at (-1.5,-0.9);
\coordinate [label=center:$+$] (B) at (1.08,-0.58);
\coordinate [label=center:$-$] (B) at (1.5,-0.9);

\draw[thick,dotted] (1.5,2) arc (36.87-90:57-90:2.5);
\draw[thick,dotted] (-1.5,2) arc (-36.87-90:-57-90:2.5);
\draw (1.5,2) arc (36.87-90:-36.87-90:2.5);
\draw[thick,dotted] (-2.48205,0.299038) arc (36.87+30:57+30:2.5);
\draw[thick,dotted] (-0.982051, - 2.29904) arc (-36.87+30:-57+30:2.5);
\draw (-2.48205,0.299038) arc (36.87+30:-36.87+30:2.5);
\draw[thick,dotted] (2.48205,0.299038) arc (180-36.87-30:180-57-30:2.5);
\draw[thick,dotted] (0.982051, - 2.29904) arc (180+36.87-30:180+57-30:2.5);
\draw (2.48205,0.299038) arc (180-36.87-30:180+36.87-30:2.5);

\coordinate [label=center: $C_0$] (B) at (0,0);

\end{scope}

\begin{scope}[shift={(12.4,-2.5)}]
\draw (0,0) circle (1.5);
\coordinate [label=center:$+$] (B) at (-1.2,0);
\coordinate [label=center:$-$] (B) at (-1.8,0);
\coordinate [label=center:$-$] (B) at (1.2,0);
\coordinate [label=center:$+$] (B) at (1.8,0);

\draw[thick,dotted] (-2,1.5) arc (36.87:57:2.5);
\draw[thick,dotted] (-2,-1.5) arc (-36.87:-57:2.5);
\draw (-2,1.5) arc (36.87:-36.87:2.5);
\draw[thick,dotted] (2,1.5) arc (180-36.87:180-57:2.5);
\draw[thick,dotted] (2,-1.5) arc (180+36.87:180+57:2.5);
\draw (2,1.5) arc (180-36.87:180+36.87:2.5);

\coordinate [label=center: $C_0$] (B) at (0,0);

\end{scope}

\end{tikzpicture}
\caption{Graphs of boundary type that contribute in order $\frac{1}{\lambda_0^3}$}
\label{fig:pic for dilaton2}
\end{figure}
We see that there are $10$ cases and we divide them in two types. Graphs of internal type are those graphs where all the edges of the face~$C_0$ are internal and graphs of boundary type are those graphs where at least one edge of the face~$C_0$ is boundary. The diagrams inside the dotted lines in the top row in Fig.~\ref{fig:pic for dilaton} are pieces of arbitrary larger graphs, while the graphs in the bottom row in Fig.~\ref{fig:pic for dilaton} are special ribbon graphs corresponding to closed Riemann surfaces of genus~$0$ and~$1$ respectively. The five pictures in Fig.~\ref{fig:pic for dilaton} already appeared in~\cite{Wit92} in the diagrammatic proof of the dilaton equation for~$\tau^c$. The contribution of these pictures in our situation is computed in exactly the same way, as in~\cite{Wit92}, and it gives the five terms in the integrand in line~\eqref{eq:t1 derivative,first line}.

Let us consider graphs of boundary type. Let us look at the first picture in Fig.~\ref{fig:pic for dilaton2}. Suppose that the face adjacent to~$C_0$ is labeled with~$\lambda_i$. Then the diagram inside the dotted lines gives
$$
\frac{1}{2}\frac{1}{\lambda_0(\lambda_0+\lambda_i)}=\frac{1}{2}\left(\frac{1}{\lambda_0^2}-\frac{\lambda_i}{\lambda_0^3}+\ldots\right).
$$
So, the coefficient of $\frac{1}{\lambda_0^3}$ is $-\frac{\lambda_i}{2}$. Note that the graph outside the dotted lines can be an arbitrary odd extended critical ribbon graph with boundary with a distinguished face labeled by~$\lambda_i$ and having a boundary edge. Now it is easy to see that the first picture in Fig.~\ref{fig:pic for dilaton2} gives the term~$-\sum\lambda_i\frac{\d I_2}{\d h_{i,i}}$ in line~\eqref{eq:t1 derivative,second line}. Consider now the second picture in Fig.~\ref{fig:pic for dilaton2}. The part inside the dotted lines gives
$$
\frac{1}{4}\frac{1}{\lambda_0^2(\lambda_0+\lambda_i)}=\frac{1}{4}\frac{1}{\lambda_0^3}+\ldots.
$$
So, the coefficient of $\frac{1}{\lambda_0^3}$ is $\frac{1}{4}$. Now we may shrink the interior of the dotted lines to a point and sum over all possible exteriors. This gives the term~$\frac{1}{2}\sum z_{i,j}\frac{\d I_2}{\d z_{i,j}}$ in line~\eqref{eq:t1 derivative,second line}. In the third picture in Fig.~\ref{fig:pic for dilaton2} the interior of the dotted lines gives
$$
\frac{1}{2}\frac{1}{\lambda_0(\lambda_0+\lambda_i)(\lambda_0+\lambda_j)}=\frac{1}{2}\frac{1}{\lambda_0^3}+\ldots.
$$
Therefore, the coefficient of $\frac{1}{\lambda_0^3}$ is $\frac{1}{2}$, and this picture corresponds to the term ~$\sum h_{i,j}\frac{\d I_2}{\d h_{i,j}}$ in line~\eqref{eq:t1 derivative,second line}. One can also easily see that the two pictures in the bottom row in Fig.~\ref{fig:pic for dilaton2} correspond to the two terms in the integrand in line~\eqref{eq:t1 derivative,third line}. Thus, formula~\eqref{eq:t1 derivative,first line} for the $t_1$ derivative is proved.

Formulas~\eqref{eq:dilaton,formula1},~\eqref{eq:dilaton,formula2} and~\eqref{eq:t1 derivative,first line} imply that the dilaton equation~\eqref{eq:dilaton equation} is equivalent to
\begin{align}
&\int_{\mcH_M\times\Mat_{N,N}(\mbC)}\left[\left(\frac{1}{3}\tr H^3-\frac{7}{6}\tr H^2\Lambda+\tr H\Lambda^2+\frac{2}{3}M^2+\frac{1}{12}\tr Z^3+\frac{1}{4}\tr Z\oZ^t-\frac{1}{2}N^2\right)I_1I_2I_3\right.\label{eq:dilaton2}\\
&\hspace{2.7cm}+I_1I_3\left(\frac{1}{3}\sum\lambda_i\frac{\d I_2}{\d\lambda_i}-\sum\lambda_i\frac{\d I_2}{\d h_{i,i}}+\frac{1}{2}\sum z_{i,j}\frac{\d I_2}{\d z_{i,j}}+\sum h_{i,j}\frac{\d I_2}{\d h_{i,j}}\right)\notag\\
&\hspace{2.7cm}\left.-\frac{2}{3}I_1I_2\sum\oz_{i,j}\frac{\d I_3}{\d\oz_{i,j}}\right]dHdZ=0.\notag
\end{align}
Using the relation
\begin{align*}
0=&\int_{\mcH_M\times\Mat_{N,N}(\mbC)}\sum\frac{\d}{\d\oz_{i,j}}\left(\oz_{i,j}I_1I_2I_3\right)dHdZ=\\
=&\int_{\mcH_M\times\Mat_{N,N}(\mbC)}\left[\left(N^2-\frac{1}{2}\tr Z\oZ^t\right)I_1I_2I_3+I_1I_3\sum\oz_{i,j}\frac{\d I_2}{\d\oz_{i,j}}+I_1I_2\sum\oz_{i,j}\frac{\d I_3}{\d\oz_{i,j}}\right]dHdZ,
\end{align*}
we see that equation~\eqref{eq:dilaton2} is equivalent to
\begin{align}
&\int_{\mcH_M\times\Mat_{N,N}(\mbC)}\left[\left(\frac{1}{3}\tr H^3-\frac{7}{6}\tr H^2\Lambda+\tr H\Lambda^2+\frac{2}{3}M^2+\frac{1}{12}\tr Z^3-\frac{1}{12}\tr Z\oZ^t+\frac{1}{6}N^2\right)I_1I_2I_3\right.\label{eq:dilaton3}\\
&\left.+I_1I_3\left(\frac{1}{3}\sum\lambda_i\frac{\d I_2}{\d\lambda_i}-\sum\lambda_i\frac{\d I_2}{\d h_{i,i}}+\frac{1}{2}\sum z_{i,j}\frac{\d I_2}{\d z_{i,j}}+\sum h_{i,j}\frac{\d I_2}{\d h_{i,j}}+\frac{2}{3}\sum\oz_{i,j}\frac{\d I_2}{\d\oz_{i,j}}\right)\right]dHdZ=0.\notag
\end{align}
Note that
$$
\left(\sum\lambda_i\frac{\d}{\d\lambda_i}+2\sum \oz_{i,j}\frac{\d}{\d\oz_{i,j}}+\sum h_{i,j}\frac{\d}{\d h_{i,j}}+\sum z_{i,j}\frac{\d}{\d z_{i,j}}\right)I_2=0.
$$
This relation simplifies~\eqref{eq:dilaton3} in the following way,
\begin{align}
&\int_{\mcH_M\times\Mat_{N,N}(\mbC)}\left[\left(\frac{1}{3}\tr H^3-\frac{7}{6}\tr H^2\Lambda+\tr H\Lambda^2+\frac{2}{3}M^2+\frac{1}{12}\tr Z^3-\frac{1}{12}\tr Z\oZ^t+\frac{1}{6}N^2\right)I_1I_2I_3\right.\label{eq:dilaton4}\\
&\hspace{2.7cm}\left.+I_1I_3\left(-\sum\lambda_i\frac{\d I_2}{\d h_{i,i}}+\frac{1}{6}\sum z_{i,j}\frac{\d I_2}{\d z_{i,j}}+\frac{2}{3}\sum h_{i,j}\frac{\d I_2}{\d h_{i,j}}\right)\right]dHdZ=0.\notag
\end{align}
Using now the relation
\begin{align*}
0=&\int_{\mcH_M\times\Mat_{N,N}(\mbC)}\sum\frac{\d}{\d z_{i,j}}\left(z_{i,j}I_1I_2I_3\right)dHdZ=\\
=&\int_{\mcH_M\times\Mat_{N,N}(\mbC)}\left[\left(N^2-\frac{1}{2}\tr Z\oZ^t+\frac{1}{2}\tr Z^3\right)I_1I_2I_3+I_1I_3\sum z_{i,j}\frac{\d I_2}{\d z_{i,j}}\right]dHdZ,
\end{align*}
we obtain that~\eqref{eq:dilaton4} is equivalent to
\begin{align}
&\int_{\mcH_M\times\Mat_{N,N}(\mbC)}\left[\left(\frac{1}{3}\tr H^3-\frac{7}{6}\tr H^2\Lambda+\tr H\Lambda^2+\frac{2}{3}M^2\right)I_1I_2I_3\right.\label{eq:dilaton5}\\
&\hspace{2.7cm}\left.+I_1I_3\left(-\sum\lambda_i\frac{\d I_2}{\d h_{i,i}}+\frac{2}{3}\sum h_{i,j}\frac{\d I_2}{\d h_{i,j}}\right)\right]dHdZ=0.\notag
\end{align}
Finally, using the relations
\begin{align*}
0=&\int_{\mcH_M\times\Mat_{N,N}(\mbC)}\sum\frac{\d}{\d h_{i,j}}\left(h_{i,j}I_1I_2I_3\right)dHdZ=\\
=&\int_{\mcH_M\times\Mat_{N,N}(\mbC)}\left[\left(M^2+\frac{1}{2}\tr H^3-\tr H^2\Lambda\right)I_1I_2I_3+I_1I_3\sum h_{i,j}\frac{\d I_2}{\d h_{i,j}}\right]dHdZ,\\
0=&\int_{\mcH_M\times\Mat_{N,N}(\mbC)}\sum\lambda_i\frac{\d}{\d h_{i,i}}\left(I_1I_2I_3\right)dHdZ=\\
=&\int_{\mcH_M\times\Mat_{N,N}(\mbC)}\left[\left(\frac{1}{2}\tr H^2\Lambda-\tr H\Lambda^2\right)I_1I_2I_3+I_1I_3\sum \lambda_i\frac{\d I_2}{\d h_{i,i}}\right]dHdZ,
\end{align*}
we see that equation~\eqref{eq:dilaton5} is true. This completes the proof of the dilaton equation.
\end{proof}


\section{Main conjecture}\label{section:main conjecture}

In this section we formulate a conjectural relation of the extended refined open partition function~$\tau^{o,ext}_N$ to the Kontsevich-Penner tau-function~$\tau_N$ from~\eqref{eq:KPmatrixmod}. In the case $N=1$ we show how to relate directly our matrix model~\eqref{eq:full refined matrix model} to the Kontsevich-Penner matrix model. We also discuss more evidence for the conjecture. In particular, we show that the conjecture is true in genus~$0$ and~$1$.

\subsection{Kontsevich-Penner matrix model and the partition function $\tau^{o,ext}_N$}

Let $T_k$, $k\ge 1$, be formal variables. Recall that the Kontsevich-Penner tau-function~$\tau_N$ is defined as a unique formal power series in the variables~$T_1,T_2,\ldots$ satisfying equation~\eqref{eq:KPmatrixmod} for each $M\ge 1$. It is not hard to see (see Section~\ref{subsubsection:genus expansion} below) that $\tau_N$ is a formal power series in $T_1,T_2,\ldots$ with the coefficients that are polynomials in $N$ with rational coefficients. Therefore, similarly to~$\tau_N^{o,ext}$, the function~$\tau_N$ is well-defined for all values of~$N$, not necessarily positive integers.

\begin{remark}
Note that in~\cite{Ale15a} our variables~$T_k$ are denoted by~$t_k$. Note also that we write the Kontsevich-Penner matrix integral in a way slightly different from~\cite{Ale15a} (see formula~(1.1) there). In order to identify formula (1.1) from~\cite{Ale15a} with the right-hand side of~\eqref{eq:KPmatrixmod}, one has to make the shift $\Phi\mapsto\Phi+\Lambda$ and then the variable change $\Phi=-H$.
\end{remark}

In \cite{Ale15a,Ale15b} the first author proved that
\begin{gather}\label{eq:conjecture for N=1}
\tau^{o,ext}=\left.\tau_1\right|_{\substack{T_{2i+1}=\frac{t_i}{(2i+1)!!},\\T_{2i+2}=\frac{s_i}{2^{i+1}(i+1)!}.}}
\end{gather}
We propose the following conjecture.
\begin{conjecture}\label{conjecture}
For any $N$ we have
\begin{gather*}
\tau^{o,ext}_N=\left.\tau_N\right|_{\substack{T_{2i+1}=\frac{t_i}{(2i+1)!!},\\T_{2i+2}=\frac{s_i}{2^{i+1}(i+1)!}.}}
\end{gather*}
\end{conjecture}

\subsection{Case $\bf N=1$}
In \cite{Ale15a,Ale15b} the relation between $\tau^{o,ext}$ and $\tau_1$ was established with the help of some properties of the integrable hierarchies. In this section we prove directly that for $N=1$ the integral representation~\eqref{eq:full refined matrix model} for the generating series of the extended refined open intersection numbers indeed coincides with the Kontsevich-Penner matrix integral~\eqref{eq:KPmatrixmod}.

Let
$$
\tilde{Z}_{M}:=\left.\frac{1}{c_{\Lambda,M}}\tau^{o,ext}\right|_{\substack{t_i=t_i(\Lambda),\\s_i=s_i(\Lambda).}}
$$
Then from (\ref{eq:extended matrix}) we have
\begin{eqnarray}\label{matint}
\tilde{Z}_{M}=
\frac{1}{2\pi}\int_{\mcH_M\times\mbC}e^{\frac{1}{6}\tr H^3-\frac{1}{2}\tr H^2\Lambda-\frac{1}{2}|z|^2+\frac{1}{6}z^3}\det\frac{\Lambda+B-H+ z}{\Lambda+B-H-z}
\det\frac{\Lambda }{B}\,dHd^2z,
\end{eqnarray}
where
$$
B:=\sqrt{\Lambda^2-\bar{z}}.
$$

Let us use the identity, valid for arbitrary formal series $f$ of two variables:
$$
\int_{\mbC} d^2z \,e^{-\frac{1}{2}|z|^2} f(\bar{z},z)=\int_{-\infty}^\infty dx \int_{-\infty}^\infty dy \,e^{ixy}f(-2iy,x).
$$
\begin{remark}
This relation can be considered as a simplest example of the more general relation between a complex matrix model and a Hermitian two-matrix model.
\end{remark}
This identity allows us to rewrite (\ref{matint}) as
$$
\tilde{Z}_{M}=\frac{1}{2\pi} \int_{\mcH_M} dH \int_{-\infty}^\infty dx \int_{-\infty}^\infty dy  \, e^{\frac{1}{6}\Tr H^3-\frac{1}{2}\Tr H^2 \Lambda +i x y+\frac{x^3}{6}}\det\frac{\Lambda +A-H+x}{\Lambda +A-H-x} \det\frac{\Lambda}{A},
$$
where
$$
A:=\sqrt{\Lambda^2+2iy}
$$
is a diagonal matrix $A=\diag(a_1,\dots,a_M)$. Let us change the variable of integration
$$
H\mapsto H+\Lambda +A.
$$
Then
$$
\frac{1}{6}H^3-\frac{1}{2}H^2 \Lambda\,\,\,\,\,\,\, \mapsto\,\,\,\,\,\,\, \frac{1}{3!}H^3+\frac{1}{2}H^2A+iyH +\frac{1}{6}(\Lambda+A)^2(A-2\Lambda)
$$
and
\begin{multline}
 \tilde{Z}_{M}=\frac{1}{2\pi}  \int_{\mcH_M} dH \int_{-\infty}^\infty dx \int_{-\infty}^\infty dy\,\times \\
\times \exp\left({\frac{1}{6}\Tr(\Lambda+A)^2(A-2\Lambda)+\frac{1}{6}\Tr H^3+\frac{x^3}{6} +i y (x+\Tr H)+\frac{1}{2}\Tr H^2 A}\right)\det\frac{H-x}{H+x}  \det\frac{\Lambda}{A}\notag.
\end{multline}

The Harish-Chandra-Itzykson-Zuber formula for the unitary matrix integral, dependent on two diagonal matrices
$$
V=\diag(v_1,v_2,\dots,v_M),\,\,\,\,\,\,
W=\diag(w_1,w_2,\dots,w_M)
$$
yields
$$
\int_{{\mathcal U}_M}e^{\tr U V  {\overline U}^t W } dU=\left(\prod_{k=1}^{M-1}k!\right)\frac{\det_{i,j=1}^M e^{v_iw_j}}{\prod_{1<i<j\leq M}(v_i-v_j)(w_i-w_j)}.
$$
We use this formula to integrate out the angular variables in the integral over H. Namely, we diagonalise $H$ as
$$
H=U \diag(h_1,\dots,h_M)   {\overline U}^t,
$$
where $U$ is a unitary $M\times M$ matrix, then
\begin{multline}
 \tilde{Z}_{M}=(2\pi)^\frac{M^2-M-2}{2}  \int_{-\infty}^\infty dx \int_{-\infty}^\infty dy  \,  \det\frac{\Lambda}{A} e^{\frac{1}{6}\Tr(\Lambda+A)^2(A-2\Lambda)}\times\\
\times
\int_{-\infty}^\infty d h_1 \dots \int_{-\infty}^\infty d h_M \prod_{1\leq i <j \leq M} \frac{h_i-h_j}{(h_i+h_j)(a_i-a_j)}
e^{\sum_{i=1}^M (\frac{1}{6}h_i^3+iy h_i+\frac{1}{2}h_i^2 a_i)+\frac{x^3}{6} +i y x}
\prod_{i=1}^M\frac{h_i-x}{h_i+x}.\notag
\end{multline}

We can consider $x$ as an additional eigenvalue of the Hermitian $(M+1)\times(M+1)$ matrix, which we denote it by $\tilde{\Phi}$. For the diagonal $(M+1)\times (M+1)$ matrix
$$
\tilde{A}=\diag (a_1,a_2,\dots,a_M,0),
$$
from the Harish-Chandra-Itzykson-Zuber formula it follows that
\begin{eqnarray*}
  \tilde{Z}_{M}=\frac{ \det \Lambda}{(2\pi)^{M+1}}
  \int_{-\infty}^\infty dy  \,  e^{\frac{1}{6}\Tr(\Lambda+A)^2(A-2\Lambda)}
\int_{\mcH_{M+1}} \,e^{\Tr \left(\frac{\tilde{\Phi}^3}{6}+\frac{\tilde{\Phi}^2\tilde{A}}{2}+iy \tilde{\Phi}\right)}d \tilde{\Phi} .
\end{eqnarray*}

Now we shift
$$
\tilde{\Phi} \mapsto \tilde{\Phi} -\tilde{A}
$$
so that
$$
\Tr\frac{1}{6}(\Lambda+A)^2(A-2\Lambda)+\Tr \left(\frac{\tilde{\Phi}^3}{6}+\frac{\tilde{\Phi}^2\tilde{A}}{2}+iy \tilde{\Phi}\right)\,\,\,\,\,\,\, \mapsto \,\,\,\,\,\,\, \Tr \left(\frac{\tilde{\Phi}^3}{6} -\frac{\tilde{\Lambda}^2\tilde{\Phi}}{2}-\frac{\Lambda^3}{3}\right),
$$
where
$$
\tilde{\Lambda}:=\diag(\lambda_1,\lambda_2,\dots,\lambda_M,\sqrt{-2iy}).
$$
Then
\begin{eqnarray*}
  \tilde{Z}_{M}=\frac{ \det \Lambda}{(2\pi)^{M+1}} \, e^{-\Tr \frac{\Lambda^3}{3}}
 \int_{-\infty}^\infty dy  \,
\int_{\mcH_{M+1}} e^{\Tr \left(\frac{\tilde{\Phi}^3}{6}-\frac{\tilde{\Phi}\tilde{\Lambda}^2}{2}\right)}d \tilde{\Phi} .
\end{eqnarray*}
Since
$$
\Tr\frac{\tilde{\Phi}\tilde{\Lambda}^2}{2}=\frac{1}{2}\sum_{i=1}^M \tilde{\Phi}_{i,i} \lambda_i^2 -i\,y\,\tilde{\Phi}_{M+1,M+1}
$$
we can integrate out $y$:
\begin{eqnarray*}
 \int_{-\infty}^\infty dy\, e^{iy \tilde{\Phi}_{M+1,M+1}}=2\pi \delta (\tilde{\Phi}_{M+1,M+1}),
\end{eqnarray*}
where $\delta$ is the Dirac delta-function.
Thus
\begin{eqnarray*}
  \tilde{Z}_{M}=\frac{\det \Lambda}{(2\pi)^{M}}  \, e^{-\Tr \frac{\Lambda^3}{3}}
\int_{\mcH_{M+1}} \delta (\tilde{\Phi}_{M+1,M+1})\,e^{\Tr \left(\frac{\tilde{\Phi}^3}{6}-\frac{\tilde{\Phi} {\Lambda^*}^2}{2}\right)}d \tilde{\Phi} .
\end{eqnarray*}
Here $\Lambda^*$ is an $(M+1)\times (M+1)$ diagonal matrix
$$
\Lambda^*:=\diag(\lambda_1,\lambda_2,\dots,\lambda_m,0).
$$
Let us change the variable of integration
$$
\tilde{\Phi} \mapsto \tilde{\Phi} -\Lambda^*
$$
so that
\begin{eqnarray*}
\Tr\left(\frac{\tilde{\Phi}^3}{6}-\frac{\tilde{\Phi} {\Lambda^*}^2}{2}\right)\,\,\,\,\,\,\, \mapsto \,\,\,\,\,\,\,  \Tr\left(\frac{\tilde{\Phi}^3}{6} -\frac{\tilde{\Phi}^2\Lambda^*}{2}\right) +\Tr\frac{{\Lambda}^3}{3}
\end{eqnarray*}
and
$$
 \tilde{Z}_{M}=\frac{ \det \Lambda}{(2\pi)^{M}}
\int_{\mcH_{M+1}}\delta (\tilde{\Phi}_{M+1,M+1})\,e^{\Tr \left(\frac{\tilde{\Phi}^3}{6}-\frac{\tilde{\Phi} {\Lambda^*}^2}{2}\right)}d \tilde{\Phi} .
$$

Because of the Dirac delta-function, the last integral reduces to the one over the Hermitian matrices of the form
$$
\tilde{\Phi}= \begin{pmatrix}
H & C \cr
\bar{C}^t & 0 \cr
\end{pmatrix},
$$
where $H$ is an $M\times M$ Hermitian matrix and $C$ is a complex vector. Since
\begin{eqnarray*}
\Tr \tilde{\Phi}^3 =\Tr {H}^3+3\bar{C}^t H C
\end{eqnarray*}
and
\begin{eqnarray*}
\Tr \tilde{\Phi}^2\Lambda^*=\Tr H^2\Lambda + \bar{C}^t \Lambda C
\end{eqnarray*}
we have
\begin{align*}
\tilde{Z}_{M}&=\frac{\det \Lambda}{(2\pi)^{M}}
\int_{\mcH_{M}\times\mbC^M}e^{\Tr \left(\frac{{H}^3}{6}-\frac{{H}^2 {\Lambda}}{2}\right)-\frac{1}{2}\bar{C}^t (\Lambda-{H}) C} {d{H}} \prod_{i=1}^M d C_i \\
&=
\int_{\mcH_M}e^{\frac{1}{6}\tr H^3-\frac{1}{2}\tr H^2\Lambda}\frac{\det\Lambda}{\det(\Lambda-H)}dH.
\end{align*}

\begin{remark}
We expect that a similar argument can be applied for any positive integer $N$.
\end{remark}

\subsection{Further evidence}

\subsubsection{String and dilaton equations}

String and dilaton equations for the Kontsevich-Penner model were derived in \cite{BH12,Ale15a} (In a more general setup of the Generalized Kontsevich Model the string equation in terms of the eigenvalues of the external matrix was derived already in \cite{KMMM93}). They coincide with the equations for the extended refined open partition function, derived in Sections~\ref{subsection:string equation} and~\ref{subsection:dilaton equation}.

\subsubsection{Genus expansion}\label{subsubsection:genus expansion}

Let
\begin{align*}
&F^{KP,N}:=\log\tau_N-\left.F^c\right|_{t_i=(2i+1)!!T_{2i+1}},\\
&\<\theta_{a_1}\cdots\theta_{a_n}\>^{KP,N}:=\left.\frac{\d^n F^{KP,N}}{\d T_{a_1}\cdots\d T_{a_n}}\right|_{T_*=0},\quad n\ge 1,\quad a_1,\ldots,a_n\ge 1.
\end{align*}
Then Conjecture~\ref{conjecture} is equivalent to the equation
\begin{gather}\label{eq:conjecture2}
\<\tau_{a_1}\cdots\tau_{a_l}\sigma_{c_1}\cdots\sigma_{c_k}\>^{o,ext,N}=\frac{\<\theta_{2a_1+1}\cdots\theta_{2a_l+1}\theta_{2c_1+2}\cdots\theta_{2c_k+2}\>^{KP,N}}{\prod_i(2a_i+1)!!\prod_j 2^{c_j+1}(c_j+1)!}.
\end{gather}
Let us insert genus parameters on the both sides of this equation. Let us look at the combinatorial formula~\eqref{eq:full combinatorial formula}. An elementary computation shows that for a graph $G\in\tcR^{ext}_l$ we~have
$$
-\deg\left(\prod_{e\in\Edges(G)}\lambda(e)\right)+\sum_{m\ge 0}(2m+2)\exc_m(G)=3\left(g(G)-1+l+\sum_{m\ge 0}\exc_m(G)\right),
$$
where $\deg$ denotes the degree of a rational function in $\lambda_1,\ldots,\lambda_l$. This implies that a graph $G\in\tcR^{ext}_l$ contributes only to intersection numbers $\<\prod_{i=1}^l\tau_{a_i}\prod_{j=1}^k\sigma_{c_j}\>^{o,ext,N}$ with $\sum (2a_i+1)+\sum (2c_j+2)=3(g(G)-1+l+k)$. For $g\ge 0$ define $\<\prod_{i=1}^l\tau_{a_i}\prod_{j=1}^k\sigma_{c_j}\>^{o,ext,N}_g$ to be equal to $\<\prod_{i=1}^l\tau_{a_i}\prod_{j=1}^k\sigma_{c_j}\>^{o,ext,N}$, if $\sum (2a_i+1)+\sum (2c_j+2)=3(g-1+l+k)$, and to be equal to~$0$ otherwise. Note that for a graph~$G\in\tcR^{ext}_l$ the parity of~$b(G)$ is opposite to the parity of~$g(G)$ and also~$b(G)\le g(G)+1$. Thus,
\begin{gather}\label{eq:parity for extended refined}
\<\prod\tau_{a_i}\prod\sigma_{c_j}\>^{o,ext,N}_g\,\text{is}\,\,
\begin{cases}
\text{an odd polynomial in $N$ of degree $\le g+1$}, &\text{if $g$ is even},\\
\text{an even polynomial in $N$ of degree $\le g+1$}, &\text{if $g$ is odd}.
\end{cases}
\end{gather}
In particular,
\begin{gather}\label{eq:genus 0,1 for extended refined}
\<\prod\tau_{a_i}\prod\sigma_{c_j}\>^{o,ext,N}_g=\<\prod\tau_{a_i}\prod\sigma_{c_j}\>^{o,ext,1}_g N^{g+1},\quad \text{for $g=0,1$}.
\end{gather}

Let us now look at the numbers $\<\theta_{a_1}\cdots\theta_{a_n}\>^{KP,N}$. For $n\ge 1$ denote by $\cR^{KP}_n$ the set of critical ribbon graphs with boundary, but with no boundary marked points and $n$ internal faces together with a bijective labeling $\alpha:\Faces(G)\stackrel{\sim}{\to}[n]$. Doing the Feynman diagram expansion of the Kontsevich-Penner matrix model~\eqref{eq:KPmatrixmod} (see~\cite{Saf16b}), one gets that
\begin{gather}\label{eq:KP combinatorial}
\sum_{a_1,\ldots,a_n\ge 1}\<\theta_{a_1}\cdots\theta_{a_n}\>^{KP,N}\prod_{i=1}^n\frac{1}{a_i \lambda_i^{a_i}}=\sum_{G\in\cR^{KP}_n}\frac{2^{e_I(G)-v_I(G)}}{|\Aut(G)|}N^{b(G)}\prod_{e\in\Edges(G)}\lambda(e),\quad n\ge 1.
\end{gather}
We see that, similarly to the intersection numbers~\eqref{eq:extended refined numbers}, the number $\<\theta_{a_1}\cdots\theta_{a_n}\>^{KP,N}$ is a polynomial in $N$ with rational coefficients. It is easy to see that a graph $G\in\cR^{KP}_n$ contributes only to intersection numbers $\<\theta_{a_1}\cdots\theta_{a_n}\>^{KP,N}$ with $\sum a_i=3(g(G)-1+n)$. So, for a non-negative integer~$g$ we define $\<\theta_{a_1}\cdots\theta_{a_n}\>^{KP,N}_g$ to be equal to $\<\theta_{a_1}\cdots\theta_{a_n}\>^{KP,N}$, if $\sum a_i=3(g-1+n)$, and to be equal to~$0$ otherwise. The combinatorial formula~\eqref{eq:KP combinatorial} immediately implies that
\begin{gather}\label{eq:parity for KP}
\<\prod\theta_{a_i}\>^{KP,N}_g\,\text{is}\,\,
\begin{cases}
\text{an odd polynomial in $N$ of degree $\le g+1$}, &\text{if $g$ is even},\\
\text{an even polynomial in $N$ of degree $\le g+1$}, &\text{if $g$ is odd}.
\end{cases}
\end{gather}
Therefore,
\begin{gather}\label{eq:genus 0,1 for KP}
\<\prod\theta_{a_i}\>^{KP,N}_g=\<\prod\theta_{a_i}\>^{KP,1}_g N^{g+1},\quad \text{for $g=0,1$}.
\end{gather}

Properties~\eqref{eq:parity for extended refined} and~\eqref{eq:parity for KP} agree with the conjectural equation~\eqref{eq:conjecture2}. Also these properties together with equation~\eqref{eq:conjecture for N=1} imply that Conjecture~\ref{conjecture} is true for $N=-1$. Equations~\eqref{eq:genus 0,1 for extended refined} and~\eqref{eq:genus 0,1 for KP} together with~\eqref{eq:conjecture for N=1} imply that the equation
\begin{gather}\label{eq:conjecture3}
\<\tau_{a_1}\cdots\tau_{a_l}\sigma_{c_1}\cdots\sigma_{c_k}\>_g^{o,ext,N}=\frac{\<\theta_{2a_1+1}\cdots\theta_{2a_l+1}\theta_{2c_1+2}\cdots\theta_{2c_k+2}\>_g^{KP,N}}{\prod_i(2a_i+1)!!\prod_j 2^{c_j+1}(c_j+1)!}
\end{gather}
is true for $g=0$ and $g=1$.

We have also checked equation~\eqref{eq:conjecture3} in several cases in genus~$2$.


\begin{thebibliography}{Ale15b}

\bibitem[Ale11]{Ale11} A. Alexandrov. {\it Matrix models for random partitions.} Nuclear Phys. B 851 (2011), no.~3, 620--650.

\bibitem[Ale15a]{Ale15a} A. Alexandrov. {\it Open intersection numbers, matrix models and MKP hierarchy}. Journal of High Energy Physics (2015), no.~3, 042, front matter+13 pp.

\bibitem[Ale15b]{Ale15b} A.~Alexandrov. {\it Open intersection numbers, Kontsevich-Penner model and cut-and-join operators}. Journal of High Energy Physics (2015), no.~8, 028, front matter+24 pp.

\bibitem[Ale16]{Ale16} A.~Alexandrov. {\it Open intersection numbers and free fields}. arXiv:1606.06712.

\bibitem[BH12]{BH12} E.~Brezin, S.~Hikami. {\it On an Airy matrix model with a logarithmic potential}. Journal of Physics. A. Mathematical and Theoretical 45 (2012), 045203.

\bibitem[BT15]{BT15} A. Buryak, R. J. Tessler. {\it Matrix models and a proof of the open analog of Witten's conjecture}. arXiv:1501.07888.

\bibitem[Bur15]{Bur15} A. Buryak. {\it Equivalence of the open KdV and the open Virasoro equations for the moduli space of Riemann surfaces with boundary}. Letters in Mathematical Physics~105 (2015), no.~10, 1427--1448.

\bibitem[Bur16]{Bur16} A. Buryak. {\it Open intersection numbers and the wave function of the KdV hierarchy}. Moscow Mathematical Journal 16 (2016), no. 1, 27--44.

\bibitem[DM69]{DM69} P.~Deligne, D.~Mumford. {\it The irreducibility of the space of curves of given genus}. Publications math\'ematiques de l'I.H.\'E.S.~36 (1969), 75--109.

\bibitem[HM98]{HM98} J.~Harris, I.~Morrison. {\it Moduli of  curves}. Graduate Texts in Mathematics, 187. Springer-Verlag, New York,~1998.

\bibitem[KMMMZ92]{KMMMZ92} S. Kharchev, A. Marshakov, A. Mironov, A. Morozov, A. Zabrodin. {\it Towards unified theory of 2d gravity}. 
Nuclear Phys. B 380 (1992), no.~1-2, 181--240.

\bibitem[KMMM93]{KMMM93} S. Kharchev, A. Marshakov, A. Mironov, A. Morozov. {\it Generalized Kontsevich model versus Toda hierarchy and discrete matrix models}. 
Nuclear Phys. B 397 (1993), no.~1-2, 339-378.

\bibitem[Kon92]{Kon92} M. Kontsevich. {\it Intersection theory on the moduli space of curves and the matrix Airy function}. Communications in Mathematical Physics~147 (1992), no. 1, 1--23.

\bibitem[Liu02]{Liu02} C.-C. M. Liu. {\it Moduli of $J$-holomorphic curves with Lagrangian boundary conditions and open Gromov-Witten invariants for an $S^1$-equivariant pair}. arXiv:math/0210257.

\bibitem[PST14]{PST14} R.~Pandharipande, J.~P.~Solomon, R.~J.~Tessler. {\it Intersection theory on moduli of disks, open KdV and Virasoro}. arXiv:1409.2191.

\bibitem[Saf16a]{Saf16a} B. Safnuk. {\it Topological recursion for open intersection numbers}. arXiv:1601.04049.

\bibitem[Saf16b]{Saf16b} B. Safnuk. {\it Combinatorial models for moduli spaces of open Riemann surfaces}. arXiv:1609.07226.

\bibitem[STa]{STa} J.~P.~Solomon, R.~J.~Tessler. {\it To appear}.

\bibitem[STb]{STb} J.~P.~Solomon, R.~J.~Tessler. {\it To appear}.

\bibitem[Str84]{Str84} K.~Strebel. {\it Quadratic differentials}. Ergebnisse der Mathematik und ihrer Grenzgebiete~(3) [Results in Mathematics and Related Areas~(3)], 5. Springer-Verlag, Berlin, 1984. xii+184 pp.

\bibitem[Tes15]{Tes15} R. J. Tessler. {\it The combinatorial formula for open gravitational descendents}. arXiv:1507.04951.

\bibitem[Tes]{Tes} R. J. Tessler. {\it To appear}.

\bibitem[Wit91]{Wit91} E. Witten. {\it Two-dimensional gravity and intersection theory on moduli space}. Surveys in differential geometry (Cambridge, MA, 1990), 243--310, Lehigh Univ., Bethlehem, PA, 1991.

\bibitem[Wit92]{Wit92} E.~Witten. {\it On the Kontsevich model and other models of two-dimensional gravity}. Proceedings of the XXth International Conference on Differential Geometric Methods in Theoretical Physics, Vol. 1, 2 (New York, 1991), 176--216, World Sci. Publ., River Edge, NJ, 1992.

\end{thebibliography}
\end{document}